\documentclass[prodmode,acmtap]{acmlarge}
\usepackage{amsmath,amssymb,cite}
\usepackage{graphicx}
\usepackage{color}
\usepackage{subcaption}
\usepackage{pbox}
\usepackage{tkz-euclide}
\usetkzobj{all}
\usepackage{tikz}
\usepackage{url}
\usepackage{todonotes}
\usepackage[bookmarks=false,colorlinks=true,urlcolor=blue, citecolor=black, linkcolor = black]{hyperref}
\usepackage{cleveref}

\newtheorem{lem}{Lemma}
\newtheorem{thm}{Theorem}
\newtheorem{defn}{Definition}
\newtheorem{coro}{Corollary}

\newtheorem{propty}{Property}
\newtheorem{rem}{Remark}

\newcommand{\E}[1]{\mathbb{E}\left[{#1}\right]}

\newcommand{\posfunc}[1]{\left( {#1}\right) ^{+}}
\newcommand{\HyperExp}{\textit{HyperExp}}
\newcommand{\SExp}{\textit{ShiftedExp}}
\newcommand{\ShiftedExp}{\textit{ShiftedExp}}
\newcommand{\Exp}{\textit{Exp}}
\newcommand{\Pareto}{\textit{Pareto}}

\DeclareMathOperator*{\argmin}{arg\,min}

\graphicspath{{Figures/}}

\crefname{equation}{}{}
\Crefname{equation}{}{}
\crefname{thm}{theorem}{theorems}
\Crefname{thm}{Theorem}{Theorems}
\crefname{clm}{claim}{claims}
\Crefname{clm}{Claim}{Claims}
\Crefname{coro}{Corollary}{Corollaries}
\Crefname{lem}{Lemma}{Lemmas}
\Crefname{sec}{Section}{Sections}
\crefname{app}{appendix}{appendices}
\Crefname{app}{Appendix}{Appendices}
\crefname{prop}{proposition}{propositions}
\Crefname{prop}{Proposition}{Propositions}
\Crefname{propty}{Property}{Properties}
\crefname{figure}{fig.}{figures}
\Crefname{figure}{Fig.}{Figures}
\crefname{defn}{definition}{definitions}
\Crefname{defn}{Definition}{Definitions}
\crefname{fact}{fact}{facts}
\Crefname{fact}{Fact}{Facts}
\crefname{appendix}{appendix}{appendices}
\Crefname{appendix}{Appendix}{Appendices}
\crefname{algo}{algorithm}{algorithms}
\Crefname{algo}{Algorithm}{Algorithms}
\crefname{algorithm}{algorithm}{algorithms}
\Crefname{algorithm}{Algorithm}{Algorithms}
\crefname{conj}{conjecture}{conjectures}
\Crefname{conj}{Conjecture}{Conjectures}
\crefname{obs}{observation}{observations}
\Crefname{obs}{Observation}{Observations}

\begin{document}

\newcommand{\totalCaches}{m}
\newcommand{\nFork}{n}
\newcommand{\nPartialFork}{r}
\newcommand{\kJoin}{k}
\newcommand{\xm}{x_m}

% ACM info
\acmVolume{2}
\acmNumber{3}
\acmArticle{1}
\articleSeq{1}
\acmYear{2016}
\acmMonth{1}

% Page heads
\markboth{G. Joshi, E. Soljanin and G. Wornell}{Efficient Redundancy Techniques for Latency Reduction in Cloud Systems}

\title{Efficient Redundancy Techniques\\for Latency Reduction in Cloud Systems}
\author{GAURI JOSHI, Carnegie Mellon University \\
EMINA SOLJANIN, Rutgers University \\
GREGORY WORNELL, Massachusetts Institute of Technology}

\begin{abstract}
In cloud computing systems, assigning a task to multiple servers and waiting for the earliest copy to finish is an effective method to combat the variability in response time of individual servers, and reduce latency. But adding redundancy may result in higher cost of computing resources, as well as an increase in queueing delay due to higher traffic load. This work helps understand when and how redundancy gives a cost-efficient reduction in latency. For a general task service time distribution, we compare different redundancy strategies in terms of the number of redundant tasks, and time when they are issued and canceled. We get the insight that the \emph{log-concavity} of the task service time creates a dichotomy of when adding redundancy helps. If the service time distribution is log-convex (i.e. log of the tail probability is convex) then adding maximum redundancy reduces both latency and cost. And if it is log-concave (i.e. log of the tail probability is concave), then less redundancy, and early cancellation of redundant tasks is more effective. Using these insights, we design a general redundancy strategy that achieves a good latency-cost trade-off for an arbitrary service time distribution. This work also generalizes and extends some results in the analysis of fork-join queues.
\end{abstract}

%\category{C.4}{Performance of Systems}{Modeling Techniques, Reliability, availability, and serviceability}%[]
%\category{G.3}{Probability and Statistics}{Queueing Theory}%[]
%
%\terms{task replication, fork-join queues}
%\keywords{performance modeling, latency-cost analysis}

%\acmformat{Joshi, G. and Soljanin, E. and Wornell, G., 2015, Efficient Redundancy Techniques for Latency Reduction in Cloud Systems.}

\begin{bottomstuff}
This work was supported in part by NSF under Grant No. CCF-1319828, AFOSR under Grant No. FA9550-11-1-0183, and a Schlumberger Faculty for the Future Fellowship. This work was presented in part at the Allerton Conference on Communication, Control and Computing 2015, and ACM Sigmetrics Mathematical Modeling and Analysis Workshop 2015. Authors' email addresses: Gauri Joshi: gaurij@andrew.cmu.edu (this author was at MIT at the time of this work); Emina Soljanin: emina.soljanin@rutgers.edu; Gregory W. Wornell: gww@mit.edu
\end{bottomstuff}

\maketitle
\section{INTRODUCTION}
\label{sec:intro}
\subsection{Motivation}
An increasing number of applications are now hosted on the cloud. Some examples are streaming (NetFlix, YouTube), storage (Dropbox, Google Drive) and computing (Amazon EC2, Microsoft Azure) services. A major advantage of cloud computing and storage is that the large-scale sharing of resources provides scalability and flexibility. However, an adverse effect of the sharing of resources is the variability in the latency experienced by the user due to queueing, virtualization, server outages etc. The problem becomes further aggravated when the computing job has several parallel tasks, because the slowest task becomes the bottleneck in job completion. Thus, ensuring fast and seamless service is a challenging problem in cloud systems. 

One method to reduce latency that has gained significant attention in recent years is the use of redundancy. In cloud computing, replicating a task on multiple machines and waiting for the earliest copy to finish can significantly reduce the latency \cite{dean_tail_2013}. Similarly, in cloud storage systems, requests to access a content can be assigned to multiple replicas, such that it is only sufficient to download one replica. However, redundancy can result in increased use of resources such as computing time, and network bandwidth. In frameworks such as Amazon EC2 and Microsoft Azure which offer computing as a service, the computing time spent on a job is proportional to the cost of renting the machines. 
\begin{table}[t]
\centering
\caption{Organization of main latency-cost analysis results presented in the rest of the paper. We fork each job into tasks at all $n$ servers (full forking), or to some subset $r$ out of $n$ servers (partial forking). A job is complete when any $k$ of its tasks are served. \label{tbl:organization}}
\begin{tabular}{|  p{4.0 cm} | p{4.7cm} | p{4.7cm} |}
\hline
& \textbf{$k=1$ (Replicated) Case} & \textbf{General $k$} \\
\hline
\textbf{Full forking to all $n$ servers}  &  \pbox[c]{4.6cm}{\vspace{0.1cm} \Cref{sec:rep_with_queueing}  \\ Comparison of strategies with and without early task cancellation \vspace{0.1cm}} & \pbox[c]{4.6cm}{\Cref{sec:coded_with_queueing} \\ Bounds on latency and cost, and the diversity-parallelism trade-off} \\
\hline
\textbf{Partial forking to $r$ out of $n$ servers} & \pbox[c]{4.6cm}{  \vspace{0.1cm} \Cref{sec:partial_fork} \\ Effect of $r$ and the choice of servers on latency and cost \vspace{0.1cm} } & \pbox[c]{4.6cm}{\Cref{sec:heuristic_algo} \\ General redundancy strategy for cost-efficient latency reduction}\\
\hline
\end{tabular}
\end{table}

\subsection{Organization of this Work}
In this work we aim to understand the trade-off between latency and computing cost, and propose efficient strategies to add redundancy. We focus on a redundancy model called the $(n,k)$ fork-join model, where a job is forked into $n$ tasks such that completion of any $k$ tasks is sufficient to finish the job. In \Cref{sec:prob_setup} we formally define this model and its variants. \Cref{sec:prev_work_contri} summarizes related previous work and our contributions. \Cref{sec:key_concepts} gives the key preliminary concepts used in this work. 

The rest of the paper studies different variants of the $(n,k)$ fork-join model in increasing order of generality, as shown in Table~\ref{tbl:organization}. In \Cref{sec:rep_with_queueing} and \Cref{sec:partial_fork} we focus on the $k=1$ (replicated) case. \Cref{sec:rep_with_queueing} considers full replication of a job at all $n$ servers, and compares different strategies of canceling redundant tasks. In \Cref{sec:partial_fork} we consider partial replication at $r$ out of $n$ servers. 

In \Cref{sec:coded_with_queueing} and \Cref{sec:heuristic_algo}, we move to the general $k$ case, which requires a significantly different style of analysis than the $k=1$ case. In \Cref{sec:coded_with_queueing} we consider full forking to all $n$ servers, and determine bounds on latency and cost, generalizing some of the fundamental work on fork-join queues. For partial forking, we propose a general redundancy strategy in \Cref{sec:heuristic_algo}. System designers looking for a practical redundancy strategy rather than theoretical analysis may skip ahead to \Cref{sec:heuristic_algo} after the problem setup in \Cref{sec:prob_setup}. 

Finally, \Cref{sec:conclu} summarizes the results and provides future perspectives. Properties and examples of log-concavity are given in Appendix~\ref{sec:tail_properties}. Proofs of the $k=1$ and general $k$ cases are deferred to Appendix~\ref{sec:rep_queueing_proofs} and Appendix~\ref{sec:coded_queueing_proofs} respectively.

\section{SYSTEM MODEL}
\label{sec:prob_setup}
\subsection{Fork-Join Model and its Variants}
\label{subsec:sys_model}

%Consider a distributed system with $n$ statistically identical servers. We now define the $(n,k)$ fork-join system and its variants that are studied in this paper. 

\begin{defn}[$(n,k)$ fork-join system]
\label{defn:fork_join}
Consider a distributed system with $n$ statistically identical servers. Jobs arrive to the system at rate $\lambda$, according to a Poisson process\footnote{The Poisson assumption is required only for the exact analysis and bounds on latency (equations \eqref{eqn:E_T_rep_queueing}, \eqref{eqn:E_T_lmbda_scaling_n}, \eqref{eqn:ET_rep_early_cancel}, \eqref{eqn:E_T_group_based}, \eqref{eqn:upper_bnd_gen}, \eqref{eqn:lower_bnd_gen}, and \eqref{eqn:ET_early_cancel_bnd}). All other results on $\E{C}$, and comparison of replication strategies in heavy traffic hold for any arrival process.}. Each job is forked into $\nFork$ tasks that join first-come first-served queues at each of the $\nFork$ servers. The job is said to be complete when any $\kJoin$ tasks are served. At this instant, all remaining tasks are canceled and abandon their respective queues immediately. %\footnote{In practice, there will be a non-zero delay in cancelling the redundant tasks. It can be accounted for by adding a constant term to the overall latency.} 
\end{defn}

After a task of the job reaches the head of its queue, the time taken to serve it can be random due to various factors such as disk seek time and sharing of computing resources between multiple processes. We model this service time by a random variable $X > 0$, with cumulative distribution function (CDF) $F_X(x)$. The tail distribution (inverse CDF) of $X$ is denoted by $\bar{F}_X(x) = \Pr(X>x)$. We use $X_{k:n}$ to denote the $k^{th}$ smallest of $n$ i.i.d.\ random variables $X_1, X_2, \dots , X_n$. 

We assume that the service time $X$ is i.i.d.\ across tasks and servers. Thus, if a task is replicated at two different servers, the service times of the replicas are independent and identically distributed. Dependence of service time on the task itself can be modeled by adding a constant $\Delta$ to $X$. More generally, $\Delta$ may be a random variable. Although we do not consider this case here, the results in this paper (particularly \Cref{sec:rep_with_queueing}) can be extended to consider correlated service times. 

Fig.~\ref{fig:fork_join_queue} illustrates the $(3,2)$ fork-join system. The job exits the system when any $2$ out of $3$ tasks are complete. The $k=1$ case corresponds to a replicated system where a job is sent to all $n$ servers and we wait for one of the replicas to be served. The $(n,k)$ fork-join system with $k > 1$ can serve as a useful model to study content access latency from an $(n,k)$ erasure coded distributed storage system. Approximate computing applications that require only a fraction of tasks of a job to be complete can also be modeled using the $(n,k)$ fork-join system. 

\begin{figure}[t]
    \begin{minipage}[t]{0.48\linewidth}
    \centering
    \includegraphics[width=3.5in]{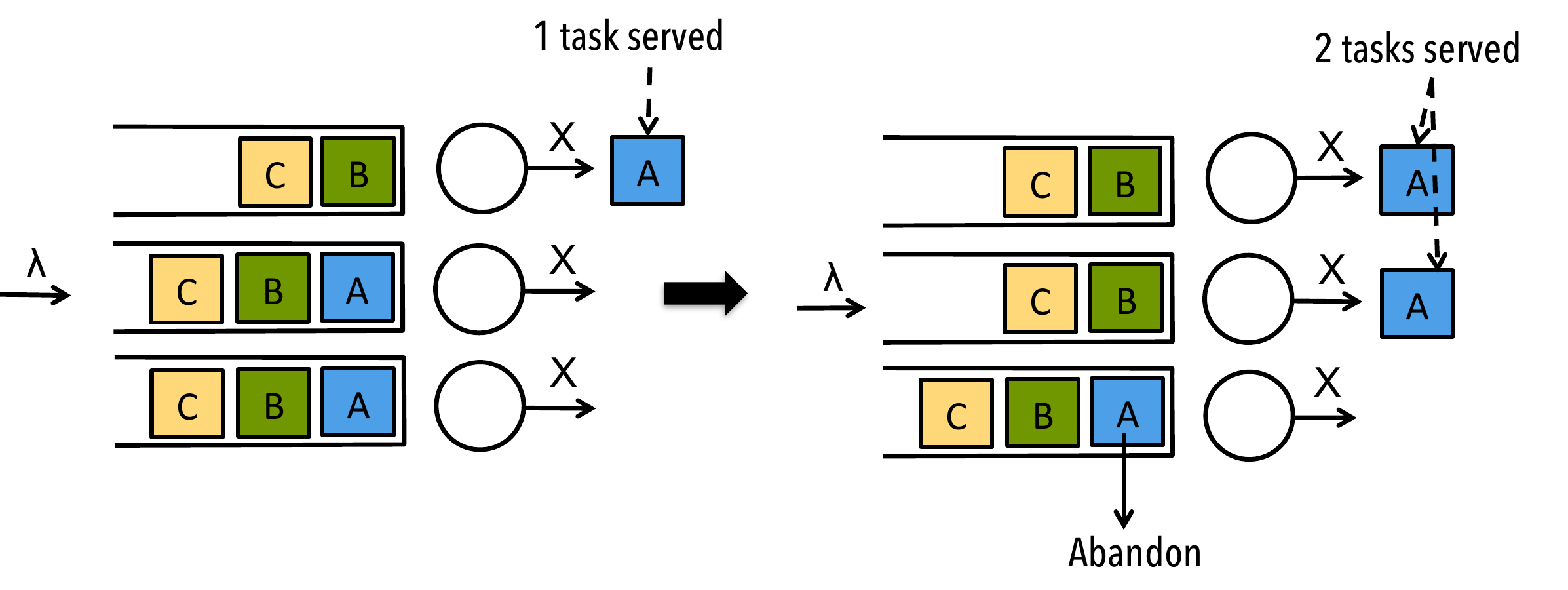}
    \caption{The $(3,2)$ fork-join system. When any $2$ out of $3$ tasks of a job are served (as seen for Job A on the right), the third task abandons its queue and the job exits the system.\label{fig:fork_join_queue}}
    \end{minipage}
    \hspace{0.04\linewidth}
    \begin{minipage}[t]{0.48\linewidth}
	\centering
	\includegraphics[width=3.3in]{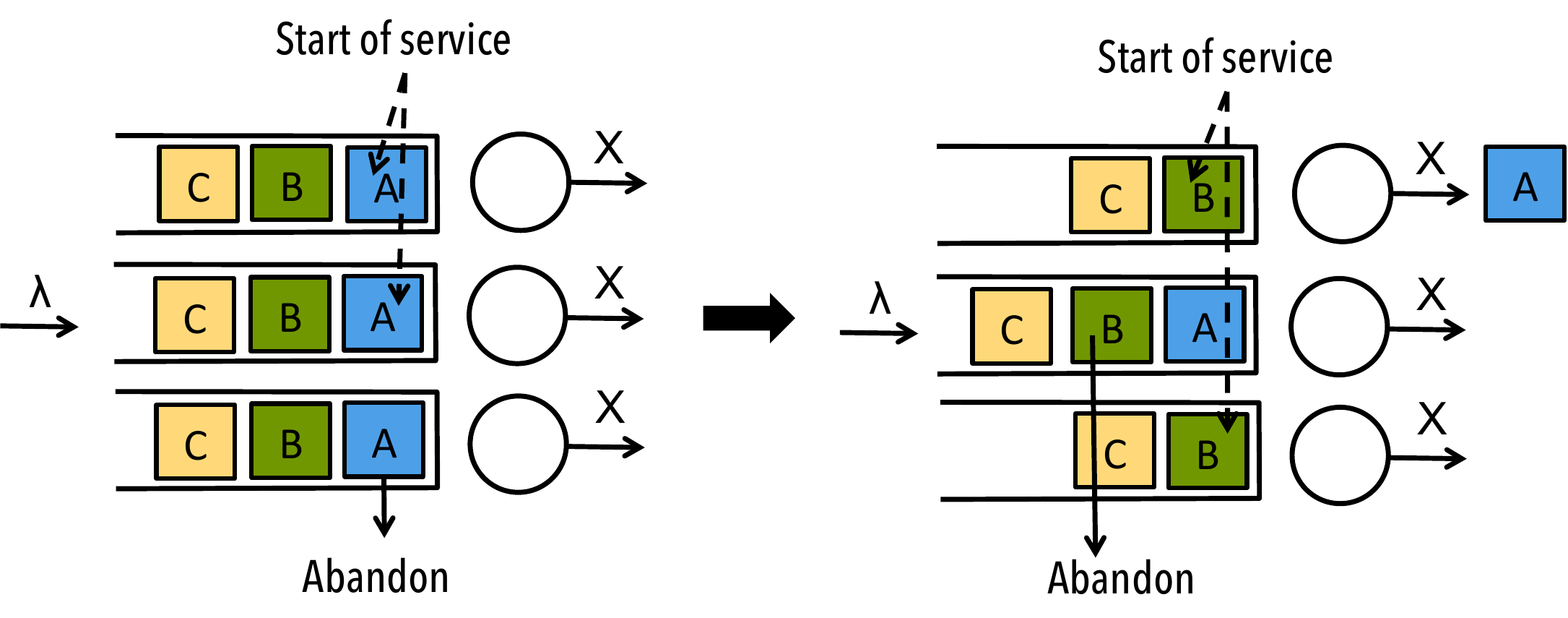}
	\caption{The $(3,2)$ fork-early-cancel system. When any $2$ out of $3$ tasks of a job are in service, the third task abandons (seen for Job A on the left, and Job B on the right). \label{fig:fork_early_cancel}}
    \end{minipage}
\end{figure}

We consider the following two variants of this system, which could save the amount of redundant time spent by the servers of each job.
\begin{enumerate}
\item \emph{$(n,k)$ fork-early-cancel system}: Instead of waiting for $k$ tasks to finish, the redundant tasks are canceled when any $k$ tasks reach the heads of their queues and start service. If more than $k$ tasks start service simultaneously, we retain any $k$ chosen uniformly at random. Fig.~\ref{fig:fork_early_cancel} illustrates the $(3,2)$ fork-early-cancel system. In \Cref{sec:rep_with_queueing} we compare the $(n,k)$ systems with and without early cancellation. 

\item  \emph{$(n, r, k)$ partial fork-join system:} Each incoming job is forked into $r > k$ out of the $n$ servers. When any $k$ tasks finish service, the redundant tasks are canceled immediately and the job exits the system. The $r$ servers can be chosen according to different scheduling policies such as random, round-robin, least-work-left (see \cite[Chapter~24]{mor_book} for definitions) etc.  
In \Cref{sec:partial_fork} we develop insights into the best choice of $r$, and the scheduling policy.
\end{enumerate} 

Other variants of the fork-join system include a combination of partial forking and early cancellation, or delaying invocation of some of the redundant tasks. Although not studied in detail here, our analysis techniques can be extended to these variants. In \Cref{sec:heuristic_algo} we propose a general redundancy strategy that is a combination of partial forking and early cancellation. 

\subsection{Latency and Cost Metrics}
\label{subsec:perf_metrics}

We now define the metrics of the latency and resource cost whose trade-off is analyzed in the rest of the paper.

\begin{defn}[Latency]
The latency $T$ is defined as the time from the arrival of a job until it is served. In other words, it is the response time experienced by the job.
\end{defn}

In this paper we focus on analyzing the expected latency $\E{T}$. Although $\E{T}$ is a good indicator of the average behavior, system designers are often interested in the tail $\Pr(T>t)$ of the latency. For many queueing problems, determining the distribution of response time $T$ requires the assumption of exponential service time. In order to consider arbitrary, non-exponential service time distribution $F_X$, we settle for analyzing $\E{T}$ here.

\begin{defn}[Computing Cost]
The computing cost $C$ is the total time spent by the servers serving a job, not including the time spent in the queue. 
\end{defn}

In computing-as-a-service frameworks, the computing cost is proportional to money spent on renting machines to run a job on the cloud\footnote{Although we focus on this cost metric, we note that redundancy also results in a network cost of making Remote-Procedure Calls (RPCs) made to assign tasks of a job, and cancel redundant tasks. It is proportional to the number of servers each job is forked to, which is $n$ for the $(n,k)$ fork-join model described above. In the context of distributed storage, redundancy also results in increased use of storage space, proportional to $\nFork/\kJoin$. The trade-off between delay and storage is studied in \cite{gauri_yanpei_emina_allerton,gauri_yanpei_emina_jsac}.}.

\section{PREVIOUS WORK AND MAIN CONTRIBUTIONS}
\label{sec:prev_work_contri}

\subsection{Related Previous Work}

\textit{\textbf{Systems Work}}:
The use of redundancy to reduce latency is not new. One of the earliest instances is the use of multiple routing paths \cite{maxemchuk2} to send packets in networks; see \cite[Chapter~7]{kabatiansky_krouk_semenov} for a detailed survey of other related work. A similar idea has been studied \cite{vulimiri_low_2013} in the context of DNS queries. In large-scale cloud computing frameworks, several recent works in systems \cite{map_reduce,ananthanarayanan_effective_2013,zaharia_sparrow} explore straggler mitigation techniques where redundant replicas of straggling tasks are launched to reduce latency. Although the use of redundancy has been explored in systems literature, there is little work on the rigorous analysis of how it affects latency, and in particular the cost of resources. Next we review some of that work. 

\textit{\textbf{Exponential Service Time}}:
The $(n,k)$ fork-join system was first proposed in \cite{gauri_yanpei_emina_allerton,gauri_yanpei_emina_jsac} to analyze content download latency from erasure coded distributed storage. These works consider that a content file coded into $n$ chunks can be recovered by accessing any $k$ out of the $n$ chunks, where the service time $X$ of each chunk is exponential. Even with the exponential assumption analyzing the $(n,k)$ fork-join system is a hard problem. It is a generalization of the $(n,n)$ fork-join system, which was actively studied in queueing literature \cite{flatto1984two,nelson_tantawi,varki_merc_chen} around two decades ago. 

Recently, an analysis of latency with heterogeneous job classes for the replicated ($k=1$) case with distributed queues is presented in \cite{gardner_sigmetrics_2015}. Other related works include \cite{mds_queue,tandon_paper,vaneet_paper,swanand_isit_2015}. A common thread in all these works is that they also assume exponential service time. 

\textit{\textbf{General Service Time}}:
Few practical systems have exponentially distributed service time. For example, studies of download time traces from Amazon S3 \cite{docomo_1,docomo_2} indicate that the service time is not exponential in practice, but instead a shifted exponential. For service time distributions that are `new-worse-than-used' \cite{cao_nbu_1991}, it is shown in \cite{koole_righter_2008} that it is optimal to replicate a job at all servers in the system. The choice of scheduling policy for new-worse-than-used (NWU) and new-better-than-used (NBU) distributions is studied in \cite{righter_job_rep_2009,shah_when_2013,sun_shroff}. The NBU and NWU notions are closely related to the log-concavity of service time studied in this work. 

\textit{\textbf{The Cost of Redundancy}}:
If we assume exponential service time then redundancy does not cause any increase in cost of server time. But since this is not true in practice, it is important to determine the cost of using redundancy. Simulation results with non-zero fixed cost of removal of redundant requests are presented in \cite{shah_when_2013}. The expected computing cost $\E{C}$ spent per job was previously considered in \cite{wang_efficient_2014,sigmetrics_arxiv_2015} for a distributed system without considering queueing of requests. In \cite{gauri_mama_2015} we presented an analysis of the latency and cost of the $(n,k)$ fork-join with and without early cancellation of redundant tasks.

\begin{table}[t]
\centering
\caption{Latency-optimal and cost-optimal redundancy strategies for the $k=1$ (replicated) case. `Canceling redundancy early' means that instead of waiting for any $1$ task to finish, we cancel redundant tasks as soon as any $1$ task begins service. \label{tbl:result_summary}}
\begin{tabular}{| p{3.0cm} | p{3.8cm} | p{2.2 cm} | p{2.5 cm} | p{2.5 cm} | }
\hline
 & \multicolumn{2}{c|} { \textbf{Log-concave service time} } & \multicolumn{2}{c|} { \textbf{Log-convex service time} }\\
\hline
 & \textbf{ Latency-optimal} & \textbf{Cost-optimal} & \textbf{Latency-optimal} & \textbf{Cost-optimal} \\
\hline
\vspace{0.05cm} \textbf{Cancel redundancy early or keep it?} \vspace{0.05cm}  
& \vspace{0.05cm} \pbox[c]{3.75cm}{Low load: Keep Redundancy, \\ High load: Cancel early} \vspace{0.05cm} & \vspace{0.05cm} Cancel early & \vspace{0.05cm} Keep Redundancy & \vspace{0.05cm} Keep Redundancy\\
\hline
\vspace{0.05cm} \textbf{Partial forking to $r$ out of $n$ servers} 
&  \vspace{0.05cm} \pbox[c]{3.7cm}{Low load: $r=n$ (fork to all), \\ High load: $r = 1$ (fork to one)} \vspace{0.1cm} & \vspace{0.05cm} $r=1$ & \vspace{0.05cm}  $r=n$ (fork to all)& \vspace{0.05cm} $r=n$ (fork to all)\\
\hline
\end{tabular}
\vspace{0.5cm}
\end{table}

\subsection{Main Contributions}
The main differences between this and previous works are: 1) we consider a general service time distribution, instead of exponential service time and, 2) we analyze the impact of redundancy on the latency, as well as the computing cost (total server time spent per job). Incidentally, our computing cost metric $\E{C}$ also serves as a powerful tool to compare different redundancy strategies under high load. 

The latency-cost analysis of the fork-join system and its variants gives us the insight that the log-concavity (and respectively, the log-convexity) of $\bar{F}_X$, the tail distribution of service time, is a key factor in choosing the redundancy strategy. Here are some examples, which are also summarized in Table~\ref{tbl:result_summary}.

\begin{itemize}
\item By comparing the $(n,1)$ systems (fork to $n$, wait for any $1$) with and without early cancellation, we can show that early cancellation of redundancy can reduce both latency and cost for log-concave $\bar{F}_X$, but it is not effective for log-convex $\bar{F}_X$. 
\item For the $(n,r,1)$ partial-fork-join system (fork to $r$ out of $n$, wait for any $1$), we can show that forking to more servers (larger $r$) is both latency and cost optimal for log-convex $\bar{F}_X$. But for log-concave $\bar{F}_X$, larger $r$ reduces latency only in the low traffic regime, and always increases the computing cost. %Further, we also show that delayed invocation of redundant tasks is better when $\bar{F}_X$ is log-concave $\bar{F}_X$, but is ineffective when $\bar{F}_X$ is log-convex.
\end{itemize}

Using these insights we also develop a general redundancy strategy to decide how many servers to fork to, and when to cancel the redundant tasks, for an arbitrary service time that may be neither log-concave nor log-convex. %In \Cref{sec:Secti we consider this general $k$ case and determine bounds on latency and cost, generalizing some of the fundamental work on fork-join queues.

\section{PRELIMINARY CONCEPTS}
\label{sec:key_concepts}
We now present some preliminary concepts that are vital to understanding the results presented in the rest of the paper.

\subsection{Using $\E{C}$ to Compare Systems}
\label{subsec:using_E_C}

Since the cost metric $\E{C}$ is the expected time spent by servers on each job, higher $\E{C}$ implies higher expected waiting time for subsequent jobs. Thus, $\E{C}$ can be used to compare the latency with different redundancy policies in the heavy traffic regime. In particular, we compare policies that are symmetric across the servers, defined formally as follows.
\begin{defn}[Symmetric Policy]
\label{defn:symmetric_forking}
With a symmetric scheduling policy, the tasks of each job are forked to one or more servers such that the expected task arrival rate is equal across all the servers.
\end{defn}

Most commonly used policies: random, round-robin, join the shortest queue (JSQ) etc.\ are symmetric across the $n$ servers. In \Cref{lem:capacity_in_terms_of_EC}, we express the stability region of the system in terms of $\E{C}$. %\Cref{coro:high_traffic_comp} then follows because higher service capacity implies lower latency in the high load regime.

\begin{lem}[Stability Region in terms of $\E{C}$]
\label{lem:capacity_in_terms_of_EC}
A system of $n$ servers with a symmetric redundancy policy is stable, that is, the mean response time $E[T] < \infty$, only if the arrival rate $\lambda$ (with any arrival process) satisfies
\begin{align}
\lambda < \frac{n}{\E{C}} \label{eqn:capacity_in_terms_of_EC}
\end{align}
Thus, the maximum arrival rate that can be supported is $\lambda_{max} = n/\E{C}$, where $\E{C}$ depends on the redundancy scheduling policy.
\end{lem}

\begin{proof}[of \Cref{lem:capacity_in_terms_of_EC}]
For a symmetric policy, the mean time spent by each server per job is $\E{C}/n$. Thus the server utilization is $\rho = \lambda \E{C}/n$. By the server utilization version of Little's Law, $\rho$ must be less than $1$ for the system to be stable. The result follows from this.
\end{proof}

\begin{defn}[Service Capacity $\lambda^{*}_{max}$]
\label{defn:serv_capacity}
The service capacity of the system $\lambda^{*}_{max}$ is the maximum achievable $\lambda_{max}$ over all symmetric policies.
\end{defn}

From \Cref{lem:capacity_in_terms_of_EC} and \Cref{defn:serv_capacity} we can infer \Cref{coro:high_traffic_comp} below. 

\begin{coro}
\label{coro:high_traffic_comp}
The redundancy strategy that minimizes $\E{C}$ results in the lowest $\E{T}$ in the heavy traffic regime ($\lambda \rightarrow \lambda^{*}_{max}$).  
\end{coro}

Note that as $\lambda$ approaches $\lambda^{*}_{max}$, the expected latency $\E{T} \rightarrow \infty$ for all strategies whose $\lambda_{max} < \lambda^{*}_{max}$.

\subsection{Log-concavity of $\bar{F}_X$}
If we fork a job to all $r$ idle servers and wait for any $1$ copy to finish, the expected computing cost $\E{C} = r\E{X_{1:r}}$, where $X_{1:r} = \min(X_1, X_2, \dots, X_r)$, the minimum of $r$ i.i.d. realizations of random variable $X$. The behavior of this cost function depends on whether the tail distribution $\bar{F}_X$ of service time is `log-concave' or `log-convex'. Log-concavity of $\bar{F}_X$ is defined formally as follows.

\begin{defn}[Log-concavity and log-convexity of $\bar{F}_X$]
\label{defn:log_concave}
The tail distribution $\bar{F}_X$ is said to be log-concave (log-convex) if $ \log \Pr(X>x)$ is concave (convex) in $x$ for all $x \in [0, \infty)$. 
\end{defn}

For brevity, when we say $X$ is log-concave (log-convex) in this paper, we mean that $\bar{F}_X$ is log-concave (log-convex). \Cref{lem:r_E_X_1_r_trend} below gives how $r \E{X_{1:r}}$ varies with $r$ for log-concave (log-convex) $\bar{F}_X$. 

\begin{lem}[Expected Minimum]
\label{lem:r_E_X_1_r_trend}
If $X$ is log-concave (log-convex), then $r \E{X_{1:r}}$ is non-decreasing (non-increasing) in $r$.
\end{lem}

The proof of \Cref{lem:r_E_X_1_r_trend} can be found in Appendix~\ref{sec:tail_properties}. Note that the exponential distribution is both log-concave and log-convex, and thus $r \E{X_{1:r}}$ remains constant as $r$ varies. This can also be seen from the fact that when $X \sim \Exp(\mu)$, an exponential with rate $\mu$, $X_{1:r}$ is an exponential with rate $r \mu$. Then, $r \E{X_{1:r}} = 1/\mu$, a constant independent of $r$.

Log-concave and log-convex distributions have been studied in economics and reliability theory and have many interesting properties. Properties relevant to this work are given in Appendix~\ref{sec:tail_properties}. We refer readers to \cite{log_concave}. In \Cref{rem:memory} we highlight one key property that provides intuitive understanding of log-concavity. 

\begin{rem}
\label{rem:memory}
It is well-known that the exponential distribution is memoryless. Log-concave distributions have `optimistic memory', that is, the expected remaining service time of a task decreases with the time elapsed. On the other hand, log-convex distributions have `pessimistic memory'.
\end{rem}

Distributions with optimistic memory are referred to as `new-better-than-used' \cite{koole_righter_2008}, `light-everywhere' \cite{shah_when_2013}, or `new-longer-than-used' \cite{sun_shroff}. Log-concavity of $X$ implies that $X$ is `new-better-than-used' (see \Cref{propty:sub_super_additivity} in Appendix~\ref{sec:tail_properties} for the proof).

%The numerical results in this paper use the shifted exponential, and hyper exponential as examples of log-concave and log-convex distributions respectively. The shifted exponential, denoted by $\SExp(\Delta,\mu)$ is an exponential with rate $\mu$, plus a constant shift $\Delta \geq 0$. The hyper-exponential distribution denoted by $\HyperExp(\mu_1, \mu_2, p)$. It is a mixture of two exponentials with decay rates $\mu_1$ and $\mu_2$ respectively, where the exponential with rate $\mu_1$ occurs with probability $p$. 

A natural question is: what are examples of log-concave and log-convex distributions that arise in practice? A canonical example of a log-concave distribution is the shifted exponential distribution $\SExp(\Delta,\mu)$, which is exponential with rate $\mu$, plus a constant shift $\Delta \geq 0$, is log-concave. Recent work \cite{docomo_1,docomo_2} on analysis of content download from Amazon S3 observed that $X$ is shifted exponential, where $\Delta$ is proportional to the size of the content and the exponential part is the random delay in starting the data transfer. Another example of log-concave service time is the uniform distribution over any convex set.

Log-convex service times occur when there is high variability in service time. CPU service times are often approximated by the hyperexponential distribution, which is a mixture of two or more exponentials. In this paper we focus on mixtures of two exponentials with decay rates $\mu_1$ and $\mu_2$ respectively, where the exponential with rate $\mu_1$ occurs with probability $p$. We denote this distribution by $\HyperExp(\mu_1, \mu_2, p)$. If a server is generally fast (rate $\mu_1$) but it can slow down (rate $\mu_2  <\mu_1$) with probability $1-p$, then the overall service time distribution would be $X \sim \HyperExp(\mu_1, \mu_2, p)$.

Many practical systems also have service times that are neither log-concave nor log-convex. In this paper we use the Pareto distribution $\Pareto(x_m, \alpha)$ as an example of such distributions. Its tail distribution is given by,
\begin{align}
\Pr(X>x) &= \begin{cases}
\left(\frac{x_m}{x}\right)^{\alpha} &  x \geq x_m, \\
1 & \text{otherwise}.
\end{cases} \label{eqn:pareto_tail}
\end{align}

The tail distribution in \eqref{eqn:pareto_tail} is log-convex for $x \geq x_m$, but not for all $x \geq 0$ due the initial delay of $x_m$. Thus, overall the Pareto distribution is neither log-concave, nor log-convex.

\begin{rem}
Log-concave (log-convex) distributions are reminiscent of another well-known class of distributions: light (heavy) tailed distributions. Many random variables with log-concave (log-convex) $\bar{F}_X$ are light (heavy) tailed respectively, but neither property implies the other. For example, the Pareto distribution defined above is heavy tailed but is neither log-concave, nor log-convex. While the tail of a distribution characterizes how the maximum $\E{X_{n:n}}$ behaves for large $n$, log-concavity (log-convexity) of $\bar{F}_X$ characterizes the behavior of the minimum $\E{X_{1:n}}$, which is of primary interest in this work.
\end{rem}

%For example, if we fork a job to all $n$ servers and wait for any $1$ copy to finish, the expected computing cost $\E{C} = n\E{X_{1:n}}$. It can be shown that $n \E{X_{1:n}}$ is non-decreasing (non-increasing) in $n$ when $\bar{F}_X$ is log-concave (log-convex). %The exponential distribution is both log-concave and log-convex and thus yields $n \E{X_{1:n}}$ that is a constant, independent of $n$.
%
%\begin{rem}
%Log-concavity of $X$ implies that $X$ is `new-better-than-used', a notion considered in \cite{koole_righter_2008}. Other names used to refer to new-better-than-used distributions are `light-everywhere' in \cite{shah_when_2013} and `new-longer-than-used' in \cite{sun_shroff}.
%\end{rem}
%
%\begin{rem}
%Many random variables with log-concave (log-convex) $\bar{F}_X$ are light (heavy) tailed respectively, but neither property implies the other. Unlike the tail of a distribution which characterizes how the maximum $\E{X_{n:n}}$, behaves for large $n$, log-concavity (log-convexity) of $\bar{F}_X$ characterizes the behavior of the minimum $\E{X_{1:n}}$, which is of primary interest in this work.
%\end{rem}

%Some other properties of log-concavity relevant to this work are given in Appendix~\ref{sec:tail_properties}. We also refer readers to \cite{log_concave} for other interesting properties and examples of log-concave distributions. 

\subsection{Relative Task Start Times}
\label{subsec:task_start_times}

Since the tasks of the job experience different waiting times in their respective queues, they start being served at different times. The relative start times of the $n$ tasks of a job is an important factor affecting the latency and cost. We denote the relative start times by $t_1 \leq t_2 \leq \cdots \leq t_n$ where $t_1 = 0$ without loss of generality. For instance, if $n=3$ tasks start at absolute times $3$, $4$ and $7$, then their relative start times are $t_1 = 0$, $t_2 = 4-3 = 1$ and $t_3 = 7-3 = 4$. In the case of partial forking when only $r$ tasks are invoked, we can consider $t_{r+1}, \cdots t_n$ to be $\infty$. 

For the replicated case $(k=1)$, let $S$ be the time from when the earliest replica of a task starts service, until any one replica finishes. It is the minimum of $X_1+t_1 , X_2 +t_2, \cdots , X_n + t_n$, where $X_i$ are i.i.d. with distribution $F_X$. The tail distributon $\Pr(S>s)$ are given by,
\begin{align}
\Pr(S> s) &= \prod_{i=1}^{n} \Pr(X > s - t_n). \label{eqn:S_tail_dist}
\end{align}

The computing cost $C$ can be expressed in terms of $S$ and $t_i$ as follows.
\begin{align}
C &= S + \posfunc{ S - t_2} + \cdots + \posfunc{S - t_n}. \label{eqn:C_expr}
\end{align}

Using \eqref{eqn:C_expr} we get several crucial insights in the rest of the paper. For instance, in Section~\ref{sec:partial_fork} we show that when $\bar{F}_X$ is log-convex, having $t_1 = t_2 = \cdots = t_n = 0$ gives the lowest $\E{C}$. Then using \Cref{lem:capacity_in_terms_of_EC} we can infer that it is optimal to fork a job to all $n$ servers when $\bar{F}_X$ is log-convex. 

%When $k = 1$, $S$ and its tail distributon $\Pr(S>s)$ are given by,
%\begin{align}
%S &= \min(X_1 + t_1, X_2 + t_2, \cdots X_n +t_n), \\
%\Pr(S> s) &= \prod_{i=1}^{n} \Pr(X > s - t_n). \label{eqn:S_tail_dist}
%\end{align}

\section{$k=1$ CASE WITHOUT AND WITH EARLY CANCELLATION} 
\label{sec:rep_with_queueing}
In this section we analyze the latency and cost of the $(n,1)$ fork-join system, and the $(n,1)$ fork-early-cancel system defined in Section~\ref{sec:prob_setup}. %The analysis is based on identifying equivalences of these systems to the $M/G/1$ and $M/G/n$ queues respectively. 
We get the insight that it is better to cancel redundant tasks early if $\bar{F}_X$ is log-concave. On the other hand, if $\bar{F}_X$ is log-convex, retaining the redundant tasks is better.

%We analyze replicated system ($k=1$) with and without early cancellation of redundant tasks and identify when early cancellation is better. We get exact analysis for $k=1$ case by identifying an equivalence between the $(n,1)$ fork-join system and an $M/G/1$ queue.
% In \Cref{sec:coded_wo_queueing} we analyzed the latency-cost performance of a coded system without queueing. Considering queueing of tasks at the servers makes the analysis, as well as the scope of possible scheduling policies much more involved. 

\begin{figure}[t]
\centering
\includegraphics[width=3.3in]{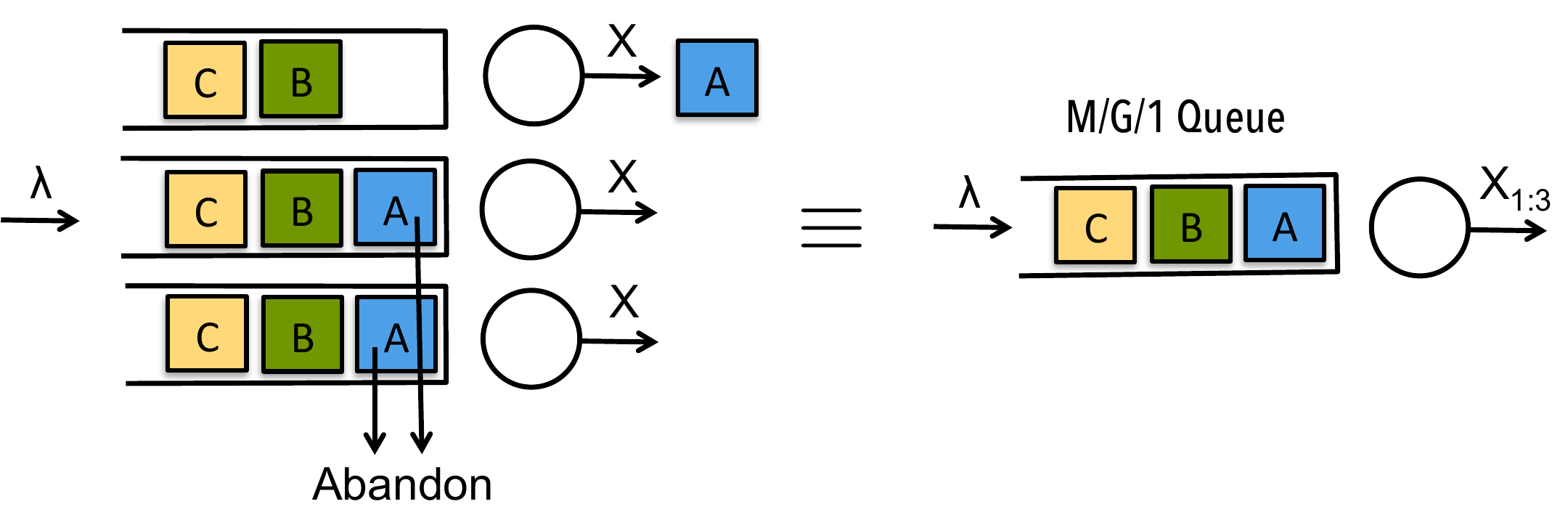}
\caption{Equivalence of the $(n,1)$ fork-join system with an $M/G/1$ queue with service time $X_{1:n}$, the minimum of $n$ i.i.d.\ random variables $X_1, X_2, \dots , X_n$. \label{fig:fork_join_mg1_eq}}
\end{figure}

\subsection{Latency-Cost Analysis}

%\begin{lem}
%\label{lem:mg1_eq}
%The latency $T$ of the $(n,1)$ fork-join system is equivalent in distribution to that of an $M/G/1$ queue with service time $X_{1:n}$. 
%\end{lem}
%
%\begin{proof}
%Consider the first job that arrives to a $(n,1)$ fork-join system when all servers are idle. Thus, the $n$ tasks of this job start service at their respective servers simultaneously. The earliest task finishes after time $X_{1:n}$, and all other tasks are canceled immediately. So, the tasks of all subsequent jobs arriving to the system also start simultaneously at the $n$ servers as illustrated in Fig.~\ref{fig:fork_join_mg1_eq}. Hence, arrival and departure events, and the latency of an $(n,1)$ fork-join system is equivalent in distribution to an $M/G/1$ queue with service time $X_{1:n}$.
%\end{proof}

\begin{thm}
\label{thm:rep_queueing}
The expected latency and computing cost of an $(n,1)$ fork-join system are given by
\begin{align}
\E{T} &= \E{T^{M/G/1}} = \E{X_{1:n}} + \frac{\lambda \E{X_{1:n}^2}}{2(1 - \lambda \E{X_{1:n}})} \label{eqn:E_T_rep_queueing} \\
\E{C} &= n \cdot \E{X_{1:n}} \label{eqn:E_C_rep_queueing} %\\
%\E{N} &= n
\end{align}
where $X_{1:n} = \min (X_1, X_2, \dots , X_n)$ for i.i.d.\ $X_i \sim F_X$.
\end{thm}

\begin{proof}
Consider the first job that arrives to a $(n,1)$ fork-join system when all servers are idle. The $n$ tasks of this job start service simultaneously at their respective servers. The earliest task finishes after time $X_{1:n}$, and all other tasks are canceled immediately. So, the tasks of all subsequent jobs arriving to the system also start simultaneously at the $n$ servers as illustrated in Fig.~\ref{fig:fork_join_mg1_eq}. Hence, arrival and departure events, and the latency of an $(n,1)$ fork-join system is equivalent in distribution to an $M/G/1$ queue with service time $X_{1:n}$.

The expected latency of an $M/G/1$ queue is given by the Pollaczek-Khinchine formula \eqref{eqn:E_T_rep_queueing}. The expected cost $\E{C} = n \E{X_{1:n}}$ because each of the $n$ servers spends $X_{1:n}$ time on the job. This can also be seen by noting that $S = X_{1:n}$ when $t_i = 0$ for all $i$, and thus by \eqref{eqn:C_expr}, $C = n X_{1:n}$.
\end{proof}

%\begin{figure}[t]
%\centering
%\includegraphics[width=3.5in]{E_T_rep_shifted_exp_vs_lambda.pdf}
%\caption{Latency versus $\lambda$ when the service distribution is Shifted exponential with $\Delta = 2$, $\mu = 0.5$. Latency decreases with $n$ for all $\lambda$. The service capacity (maximum $\lambda$ supported) is $1/\E{X_{1:n}}$, which increases with $n$.\label{fig:E_T_rep_shifted_exp_vs_lambda}}
%\end{figure}

%In \Cref{fig:E_T_rep_shifted_exp_vs_lambda} we plot the latency versus $\lambda$ for shifted exponential service time $\SExp(2, 0.5)$. 

From \eqref{eqn:E_T_rep_queueing} it is easy to see that for any service time distribution $F_X$, the expected latency $\E{T}$ is non-increasing with $n$. The behavior of $\E{C}$ follows from \Cref{lem:r_E_X_1_r_trend} as given by \Cref{coro:rep_queueing_EC} below.

\begin{coro}
\label{coro:rep_queueing_EC}
If $\bar{F}_X$ is log-concave (log-convex), then $\E{C}$ is non-decreasing (non-increasing) in $n$.
\end{coro}

\begin{figure}[t]
    \begin{minipage}[t]{0.48\linewidth}
        \centering
      \includegraphics[width=3.2in]{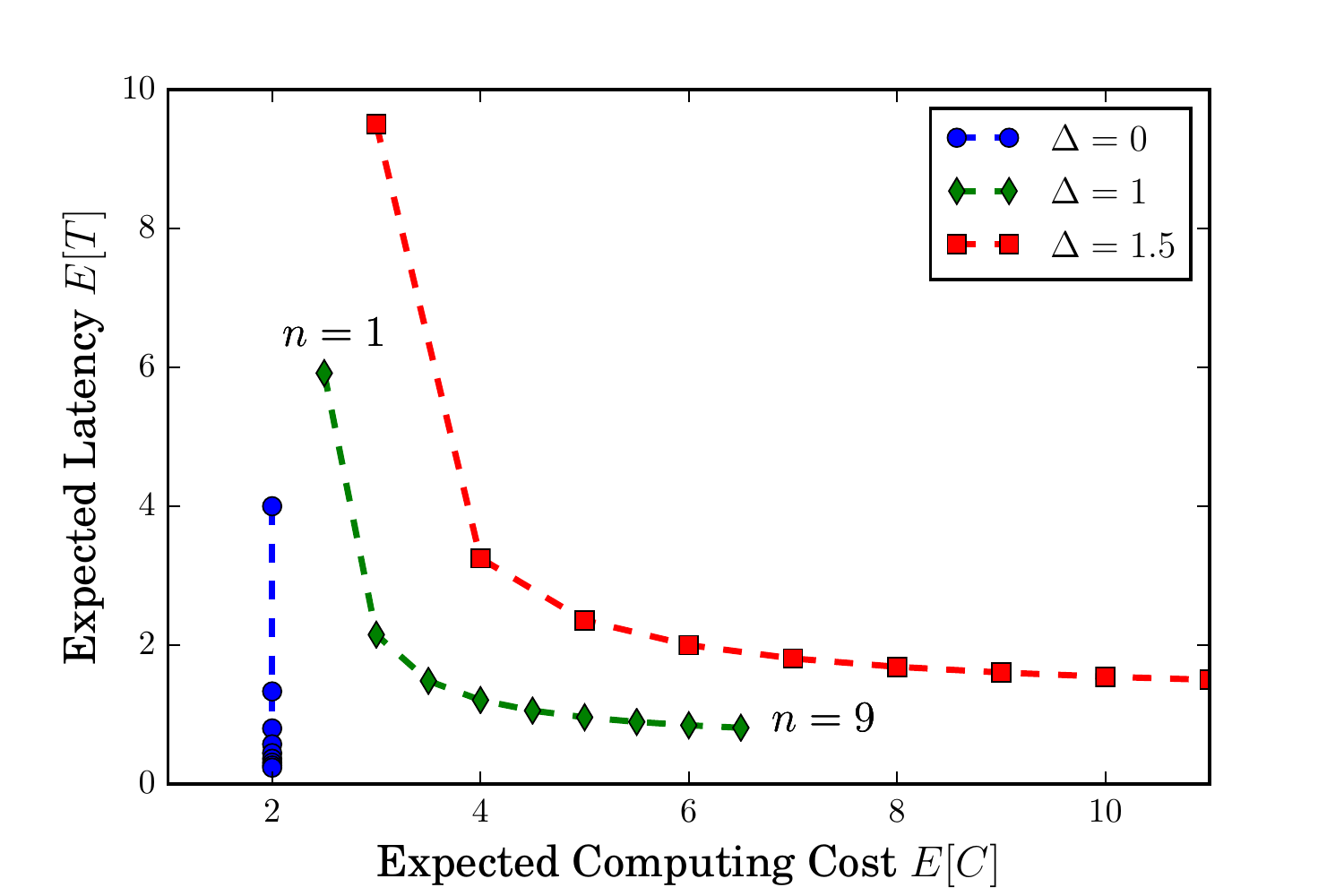}
      \caption{ The service time $X \sim \Delta + \text{Exp}(\mu)$ (log-concave), with $\mu = 0.5$, $\lambda = 0.25$. As $n$ increases along each curve, $\E{T}$ decreases and $\E{C}$ increases. Only when $\Delta =0$, latency reduces at no additional cost. \label{fig:ET_vs_EC_rep_shifted_exp_var_n}}
    \end{minipage}
    \hspace{0.04\linewidth}
    \begin{minipage}[t]{0.48\linewidth}
\centering
\includegraphics[width=3.2in]{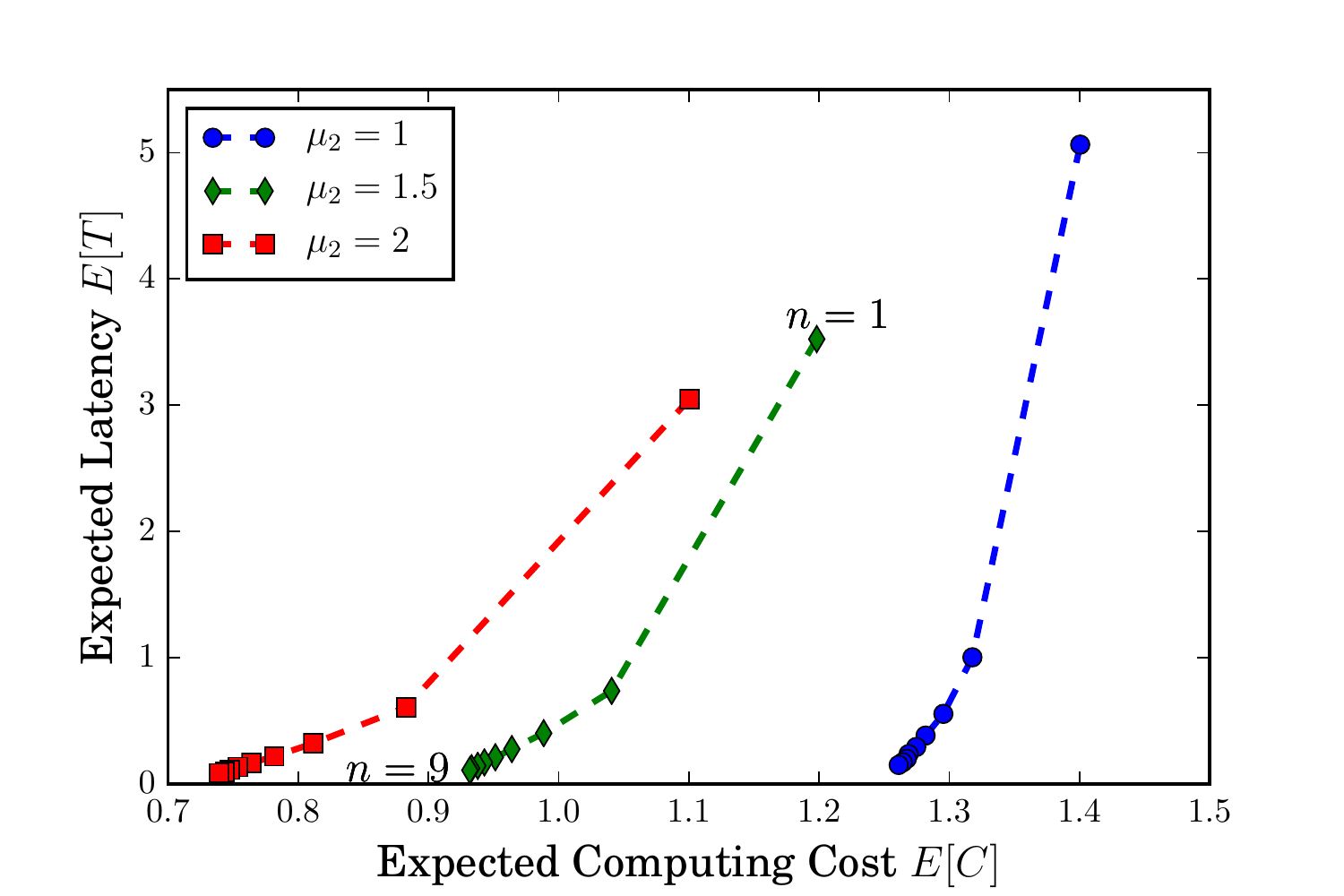}
\caption{ The service time $X \sim \HyperExp(0.4, \mu_1, \mu_2)$ (log-convex), with $\mu_1 = 0.5$, different values of $\mu_2$, and $\lambda = 0.5$.  Expected latency and cost both reduce as $n$ increases along each curve.\label{fig:ET_vs_EC_rep_hyper_exp_var_n}}
    \end{minipage}
\end{figure}

Fig.~\ref{fig:ET_vs_EC_rep_shifted_exp_var_n} and Fig.~\ref{fig:ET_vs_EC_rep_hyper_exp_var_n} show analytical plots of the expected latency versus cost for log-concave and log-convex $\bar{F}_X$ respectively. In Fig.~\ref{fig:ET_vs_EC_rep_shifted_exp_var_n}, the arrival rate $\lambda = 0.25$, and $X$ is shifted exponential $\SExp(\Delta, 0.5)$, with different values of $\Delta$. For $\Delta > 0$, there is a trade-off between expected latency and cost. Only when $\Delta = 0$, that is, $X$ is a pure exponential (which is generally not true in practice), we can reduce latency without any additional cost. In Fig.~\ref{fig:ET_vs_EC_rep_hyper_exp_var_n}, arrival rate $\lambda = 0.5$, and $X$ is hyperexponential $\HyperExp(0.4, 0.5, \mu_2)$ with different values of $\mu_2$. We get a simultaneous reduction in $\E{T}$ and $\E{C}$ as $n$ increases. The cost reduction is steeper as $\mu_2$ increases.

Instead of holding the arrival rate $\lambda$ constant, if we consider that it scales linearly with $n$, then the latency $\E{T}$ may not always decrease with $n$. In \Cref{coro:ET_lmbda_scaling_n} we study the behavior as $n$ varies.

\begin{coro}
\label{coro:ET_lmbda_scaling_n}
If the arrival rate $\lambda = \lambda_0 n$, scaling linearly with $n$, then the latency $\E{T}$ decreases with $n$ if $\bar{F}_X$ is log-convex. If $\bar{F}_X$ is log-concave then $\E{T}$ increase with $n$ in heavy traffic.
\end{coro}
\begin{proof}
If $\lambda = \lambda_0 n$, then latency $\E{T}$ in \eqref{eqn:E_T_rep_queueing} can be rewritten as
\begin{align}
\E{T} = \E{X_{1:n}} + \frac{\lambda_0 n\E{X_{1:n}^2}}{2(1 - \lambda_0 n\E{X_{1:n}})} \label{eqn:E_T_lmbda_scaling_n}
\end{align}
If $\bar{F}_X$ is log-convex then by \Cref{lem:r_E_X_1_r_trend} we know that $n \E{X_{1:n}}$ decreases with $n$. Similarly, $n \E{X_{1:n}^2}$ also decreases with $n$ (the proof follows similarly as \Cref{lem:r_E_X_1_r_trend}. Hence, we can conclude that the latency in \eqref{eqn:E_T_lmbda_scaling_n} decreases with $n$ for log-convex $\bar{F}_X$. On the other hand, if $\bar{F}_X$ is log-concave, then $n\E{X_{1:n}}$ and $n \E{X_{1:n}^2}$ increase with $n$. Thus, in the heavy traffic regime $(\lambda \rightarrow \lambda^*_{max})$, when the second term in \eqref{eqn:E_T_lmbda_scaling_n} dominates, $\E{T}$ increases with $n$.
\end{proof}

%
%\begin{figure}[t]
%\centering
%\includegraphics[width=3.2in]{ET_vs_EC_rep_shifted_exp_var_n.pdf}
%\caption{ The service time $X \sim \Delta + \text{Exp}(\mu)$ (log-concave), with $\mu = 0.5$, $\lambda = 0.25$ and different values of $\Delta$. As $n$ increases along each curve, latency decreases and cost increases. Only for $\Delta =0$ latency reduces at no additional cost. \label{fig:ET_vs_EC_rep_shifted_exp_var_n}}
%\end{figure}
%
%\begin{figure}[t]
%\centering
%\includegraphics[width=3.2in]{ET_vs_EC_rep_hyper_exp_var_n.pdf}
%\caption{ The service time $X \sim \HyperExp(0.4, \mu_1, \mu_2)$ (log-convex), with $\mu_1 = 0.5$, and different values of $\mu_2$.  Latency and cost both reduce as $n$ increases along each curve.\label{fig:ET_vs_EC_rep_hyper_exp_var_n}}
%\end{figure}

\subsection{Early Task Cancellation}
\label{subsec:rep_early_cancel}
We now analyze the $(n, 1)$ fork-early-cancel system, where we cancel redundant tasks as soon as any task reaches the head of its queue. Intuitively, early cancellation can save computing cost, but the latency could increase due to the loss of diversity advantage provided by retaining redundant tasks. Comparing it to $(n,1)$ fork-join system, we gain the insight that early cancellation is better when $\bar{F}_X$ is log-concave, but ineffective for log-convex $\bar{F}_X$.

\begin{thm}
\label{thm:rep_queueing_early_cancel}
The expected latency and cost of the $(n,1)$ fork-early-cancel system are given by
\begin{align}
\E{T} &=  \E{T^{M/G/n}}, \label{eqn:ET_rep_early_cancel} \\
\E{C} &= \E{X}, \label{eqn:rep_case_cost_early_cancel} 
\end{align}
where $T^{M/G/n}$ is the response time of an $M/G/n$ queueing system with service time $X \sim F_X$. 
\end{thm}

\begin{proof}
In the $(n,1)$ fork-early-cancel system, when any one tasks reaches the head of its queue, all others are canceled immediately. The redundant tasks help find the queue with the least work left, and exactly one task of each job is served by the first server that becomes idle. Thus, as illustrated in Fig.~\ref{fig:fork_early_mgn_eq}, the latency of the $(n,1)$ fork-early-cancel system is equivalent in distribution to an $M/G/n$ queue. Hence $\E{T} = \E{T^{M/G/n}}$ and $\E{C} = \E{X}$.
\end{proof}

\begin{figure}[t]
\centering
\includegraphics[width=3.3in]{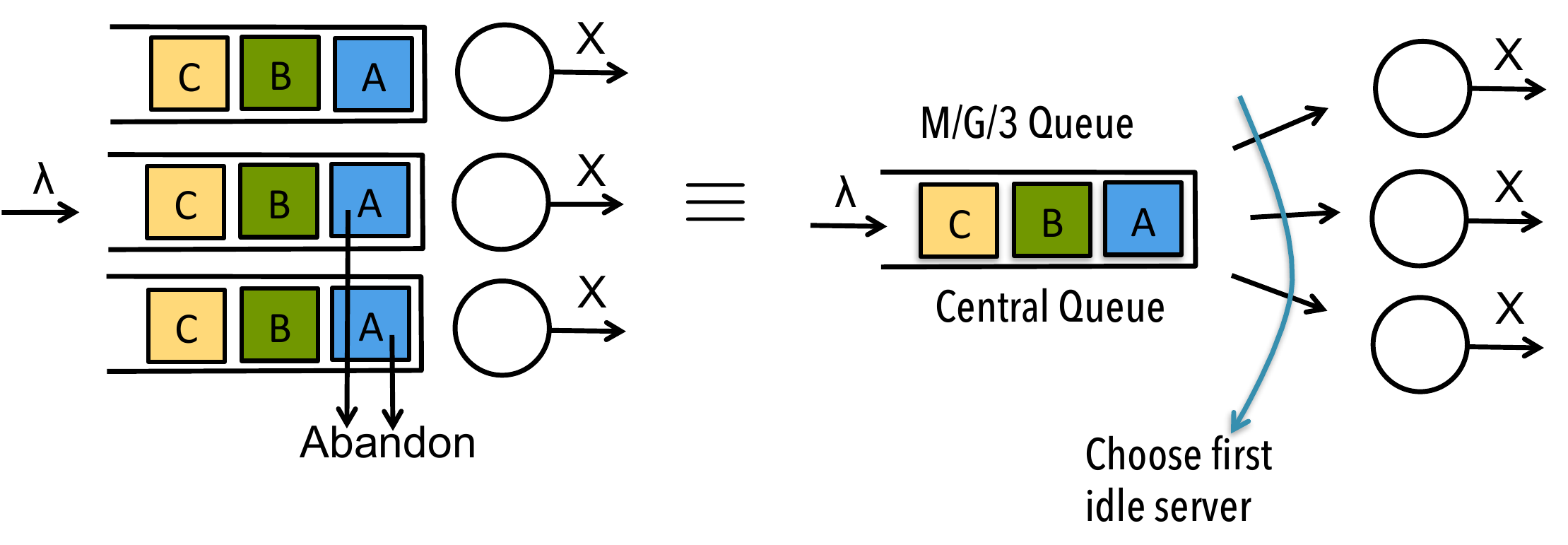}
\caption{Equivalence of the $(n,1)$ fork-early cancel system to an $M/G/n$ queue with each server taking time $X \sim F_X$ to serve task, i.i.d.\ across servers and tasks. \label{fig:fork_early_mgn_eq}}
\end{figure}

The exact analysis of mean response time $\E{T^{M/G/n}}$ has long been an open problem in queueing theory. A well-known approximation given by \cite{lee_loughton} is,
\begin{align}
 \E{T^{M/G/n}} \approx \E{X} +  \frac{ \E{X^2}}{2 \E{X}^2}\E{W^{M/M/n}}  \label{eqn:rep_case_latency_early_cancel}  
\end{align}
where $\E{W^{M/M/n}}$ is the expected waiting time in an $M/M/n$ queueing system with load $\rho = \lambda \E{X}/n$. This expected waiting time can be evaluated using the Erlang-C model \cite[Chapter~14]{mor_book}. % and is given by
%\begin{align}
%\E{W^{M/M/n}} &= \frac{ \rho (n \rho)^n }{n! (1- \rho)^2\lambda} \left( \sum_{i=0}^{n-1} \frac{(n \rho)^i}{i!} + \frac{(n \rho)^n}{n! (1-\rho)} \right)^{-1} \label{eqn:W_M_M_n_final}.
%\end{align}
A related work that studies the centralized queue model that the $(n,1)$ fork-early-cancel system is equivalent to is \cite{visschers_product_2012}, which considers the case of heterogeneous job classes with exponential service times.

Next we compare the latency and cost with and without early cancellation given by Theorem~\ref{thm:rep_queueing_early_cancel} and Theorem~\ref{thm:rep_queueing}. \Cref{coro:early_cancel_E_C_trend} below follows from \Cref{lem:r_E_X_1_r_trend}.

%By \Cref{propty:r_E_X_1_r_trend}, we know that $n \E{X_{1:n}}$ is non-decreasing with $n$ for log-concave $\bar{F}_X$. Thus, we can get the following corollary comparing the $\E{C}$ with and without early cancellation, given by \Cref{thm:rep_queueing} and \Cref{thm:rep_queueing_early_cancel}.

\begin{coro}
\label{coro:early_cancel_E_C_trend}
If $\bar{F}_X$ is log-concave (log-convex), then $\E{C}$ of the $(n,1)$ fork-early-cancel system is greater than or equal to (less than or equal to) that of $(n,1)$ fork-join system. 
\end{coro}

In the low $\lambda$ regime, the $(n,1)$ fork-join system gives lower $\E{T}$ than $(n,1)$ fork-early-cancel because of higher diversity due to redundant tasks. By \Cref{coro:high_traffic_comp}, in the high $\lambda$ regime, the system with lower $\E{C}$ has lower expected latency. 

\begin{coro}
\label{coro:early_cancel_E_T_trend}
If $\bar{F}_X$ is log-concave, early cancellation gives higher $\E{T}$ than $(n,1)$ fork-join when $\lambda$ is small, and lower in the high $\lambda$ regime. If $\bar{F}_X$ is log-convex, then early cancellation gives higher $\E{T}$ for both low and high $\lambda$.
\end{coro}

%In the high $\lambda$ regime, we can use by \Cref{clm:capacity_in_terms_of_EC} and \Cref{coro:early_cancel_E_C_trend} to imply the following result about expected latency $\E{T}$.

\begin{figure}[t]
    \begin{minipage}[t]{0.48\linewidth}
    \centering
    \includegraphics[width=3.2in]{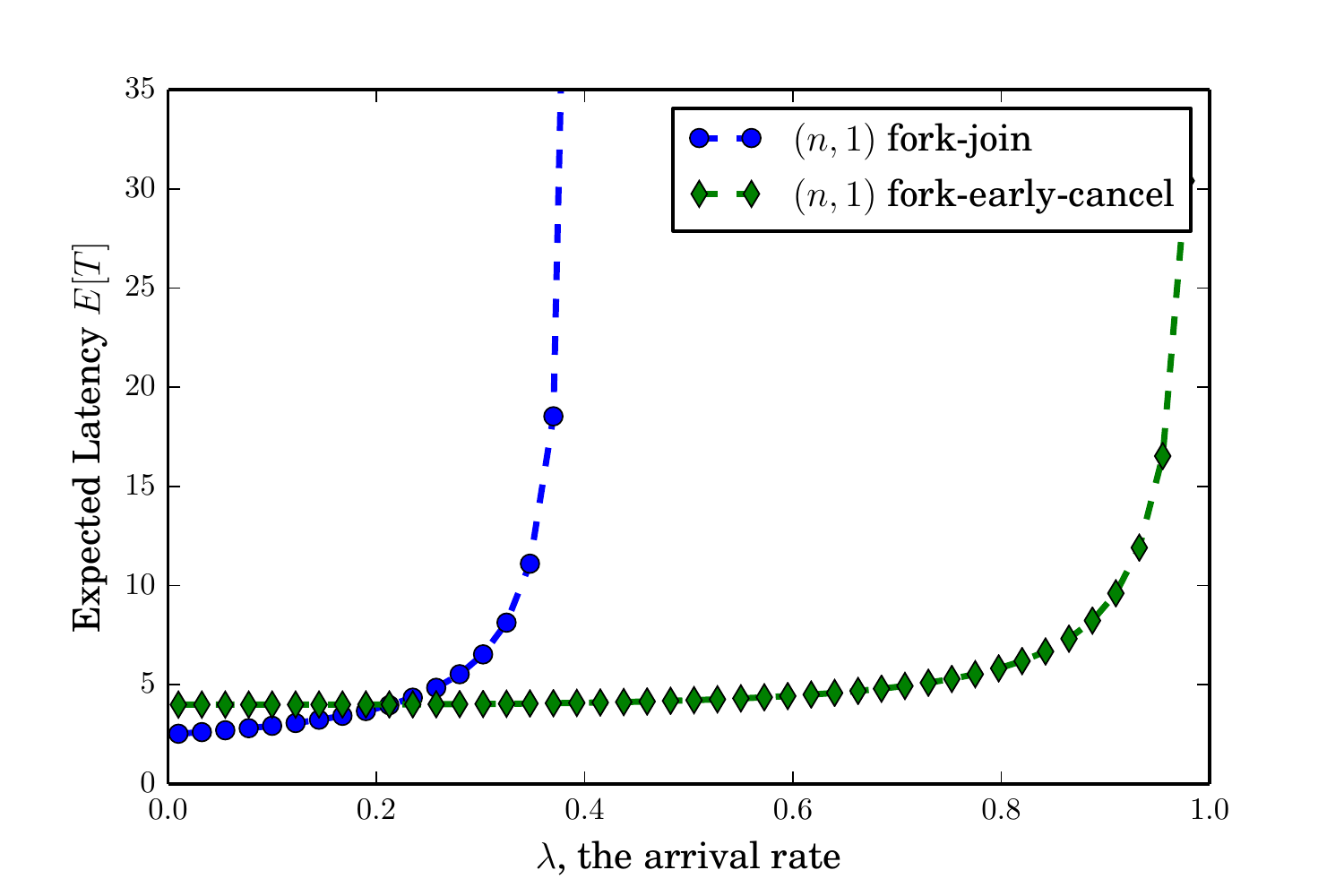}
    \caption{For the $(4,1)$ system with service time $X \sim \SExp(2, 0.5)$ which is log-concave, early cancellation is better in the high $\lambda$ regime, as given by \Cref{coro:early_cancel_E_T_trend}. \label{fig:normal_early_vs_lambda_log_concave}}
    \end{minipage}
    \hspace{0.04\linewidth}
    \begin{minipage}[t]{0.48\linewidth}
\centering
\includegraphics[width=3.2in]{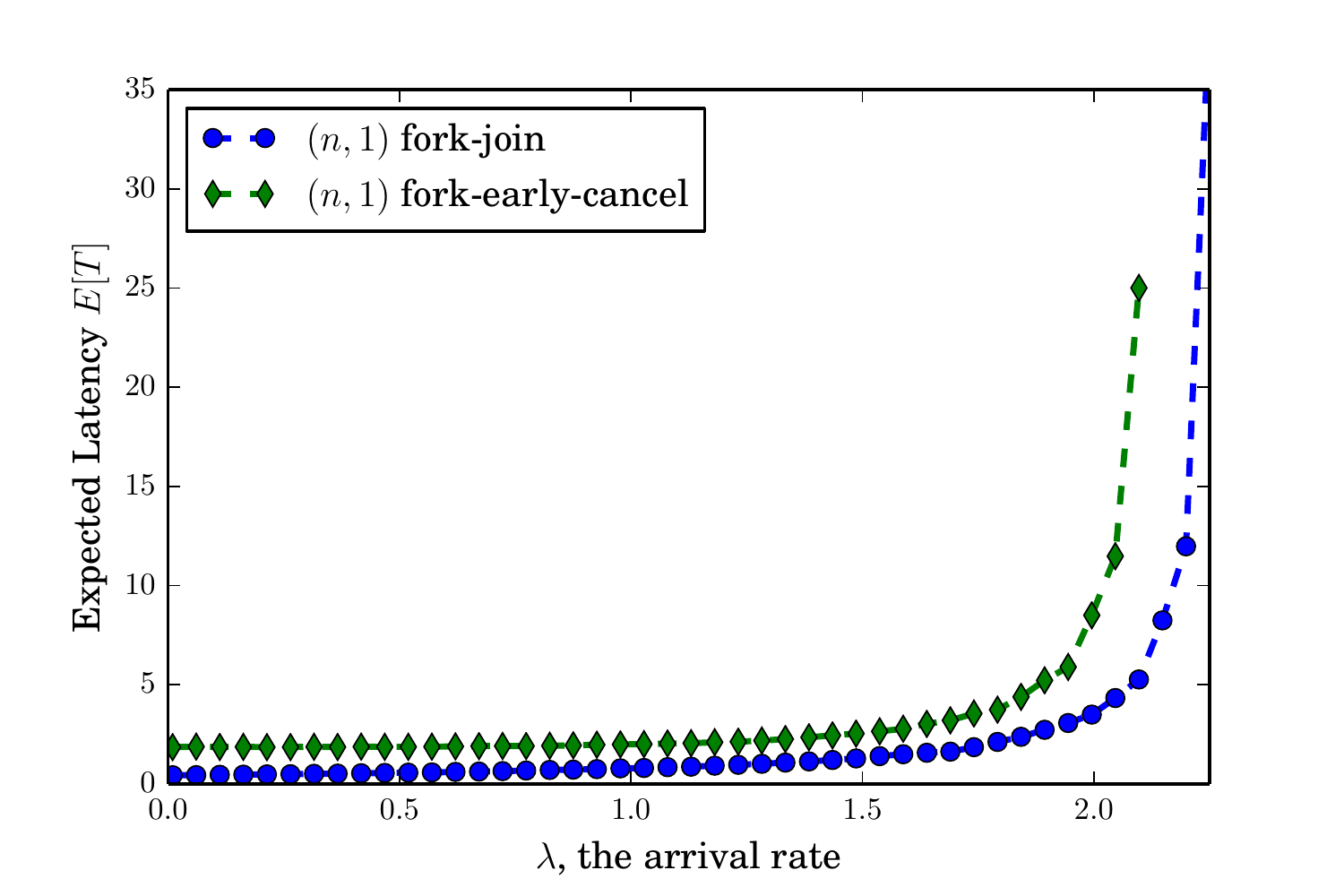}
\caption{ For the $(4,1)$ system with $X \sim \HyperExp(0.1, 1.5, 0.5)$, which is log-convex, early cancellation is worse in both low and high $\lambda$ regimes, as given by \Cref{coro:early_cancel_E_T_trend}. \label{fig:normal_early_vs_lambda_log_convex}}
    \end{minipage}
\end{figure}

%From \eqref{eqn:E_C_rep_queueing} and \Cref{clm:capacity_in_terms_of_EC} we can infer the following result about the service capacity. 
%
%\begin{coro}
%\label{coro:rep_queueing_capacity}
%The service capacity of the $(n,1)$ fork-join system is $\lambda_{\max} =1/\E{X_{1:n}}$, which is non-decreasing in $n$.
%\end{coro}
%
%This can be observed in \Cref{fig:E_T_rep_shifted_exp_vs_lambda}, where the latency goes infinity as $\lambda$ approaches $\lambda_{\max}$.

%
%\begin{figure}[t]
%\centering
%\includegraphics[width=3.2in]{E_T_normal_and_early_shifted_exp_vs_lmbda.pdf}
%\caption{For the $(4,1)$ system with service time $X \sim \SExp(1, 0.5)$, which is log-concave, early cancellation of redundant requests gives lower latency in the high $\lambda$ regime, as given by \Cref{coro:early_cancel_E_T_trend}. \label{fig:normal_early_vs_lambda_log_concave}}
%\end{figure}
%
%\begin{figure}[t]
%\centering
%\includegraphics[width=3.2in]{E_T_normal_and_early_hyper_exp_vs_lmbda.pdf}
%\caption{ For the $(4,1)$ system with $X \sim \HyperExp(0.1, 1.5, 0.5)$, which is log-convex, early cancellation of redundant requests is worse in both low and high $\lambda$ regimes, as given by \Cref{coro:early_cancel_E_T_trend}. \label{fig:normal_early_vs_lambda_log_convex}}
%\end{figure}

Fig.~\ref{fig:normal_early_vs_lambda_log_concave} and Fig.~\ref{fig:normal_early_vs_lambda_log_convex} illustrate \Cref{coro:early_cancel_E_T_trend}. Fig.~\ref{fig:normal_early_vs_lambda_log_concave} shows a comparison of $\E{T}$ with and without early cancellation of redundant tasks for the $(4,1)$ system with service time $X \sim \SExp(2,0.5)$. We observe that early cancellation gives lower $\E{T}$ in the high $\lambda$ regime. In Fig.~\ref{fig:normal_early_vs_lambda_log_convex} we observe that when $X$ is $\HyperExp(0.1, 1.5, 0.5)$ which is log-convex, early cancellation is worse for both small and large $\lambda$.

In general, early cancellation is better when $X$ is less variable (lower coefficient of variation). For example, a comparison of $\E{T}$ with $(n,1)$ fork-join and $(n,1)$ fork-early-cancel systems as $\Delta$, the constant shift of service time $\SExp(\Delta, \mu)$ varies indicates that early cancellation is better for larger $\Delta$. When $\Delta$ is small, there is more randomness in the service time of a task, and hence keeping the redundant tasks running gives more diversity and lower $\E{T}$. But as $\Delta$ increases, task service times are more deterministic due to which it is better to cancel the redundant tasks early.

\section{PARTIAL FORKING ($k=1$ CASE)}
\label{sec:partial_fork}
% %% IN THIS VERSION WE REPLACE LIGHT-EVERYWHERE BY LOG-CONCAVE 
For applications with a large number of servers $n$, full forking of jobs to all servers can be expensive in terms of the network cost of issuing and canceling the tasks. In this section we analyze the $k=1$ case of the $(n,r,k)$ fork-join system, where an incoming job is forked to some $r$ out $n$ servers and we wait for any $1$ task to finish. The $r$ servers are chosen using a symmetric policy (\Cref{defn:symmetric_forking}). Some examples of symmetric policies are:
\begin{enumerate}
\item \textit{Group-based random: } This policy holds when $r$ divides $n$. The $n$ servers are divided into $n/r$ groups of $r$ servers each. A job is forked to one of these groups, chosen uniformly at random. 
\item \textit{Uniform Random: } A job is forked to any $r$ out of $n$ servers, chosen uniformly at random.
%\item \textit{Partial Cancellation: } A job is forked to all $n$ servers. As soon as $r$ of them start service, the remaining tasks are canceled. This is equivalent to having a centralized queue and forking to the first $r$ servers that become idle.
\end{enumerate}

Fig.~\ref{fig:partial_fork} illustrates the $(4,2,1)$ partial-fork-join system with the group-based random and the uniform-random policies. In the sequel, we develop insights into the best $r$ and the choice of servers for a given service time distribution $F_X$. %The key insight we get is that in the high traffic regime, using a larger $r$ and a forking policy with results in more tasks starting at the same time is better for log-convex $\bar{F}_X$. The opposite holds for log-concave $\bar{F}_X$. 

\begin{rem}[Relation to Power-of-$r$ Scheduling]
Power-of-$r$ scheduling \cite{powerof2} is a well-known policy in multi-server systems. It chooses $r$ out of the $n$ servers at random and assigns an incoming task to the shortest queue among them. A major advantage of the power-of-$r$ policy is that even with $r << n$, the latency achieved by it is close to the join-the-shortest queue policy (equivalent to power-of-$r$ with $r = n$). 

The $(n,r,1)$ partial-fork-join system with uniform random policy also chooses $r$ queues at random. However, instead of choosing the shortest queue, it creates replicas of the task at all the queues. The replicas help find the queue with the least work left, which gives better load balancing than joining the shortest queue. But unlike power-of-$r$, servers might spend redundant time on replicas that will eventually be canceled.

%Thus, by \Cref{coro:early_cancel_E_T_trend}, we conjecture that for log-convex $\bar{F}_X$, the $(n,r,1)$ partial-fork-join system gives lower latency than power-of-$r$ scheduling.%In \cite[Fig.~4]{gauri_yanpei_emina_jsac}, we presented simulation results to show how $(n,1)$ fork-join system gives lower latency than power-of-$n$ for exponential service time. 

%In the $(n,r,1)$ partial-fork-join system with uniform random policy, we assign the job to $r$ randomly chosen queues, and wait for the earliest copy to finish. Instead, in the power-of-$r$ scheduling, a job is assigned to the shortest of $r$ randomly chosen queues. In \cite[Fig.~4]{gauri_yanpei_emina_jsac}, we presented simulation results to show how the fork-join system gives lower latency power-of-$d$ for exponential service time. From the analysis presented below, we conjecture that the $(n,r,1)$ partial-fork-join system gives lower latency than power-of-$r$ scheduling when the service distribution is log-convex. 
\end{rem}

\begin{figure}[t]
\centering
\begin{subfigure}[t]{0.5\linewidth}
    \centering
   \includegraphics[width=1.80in]{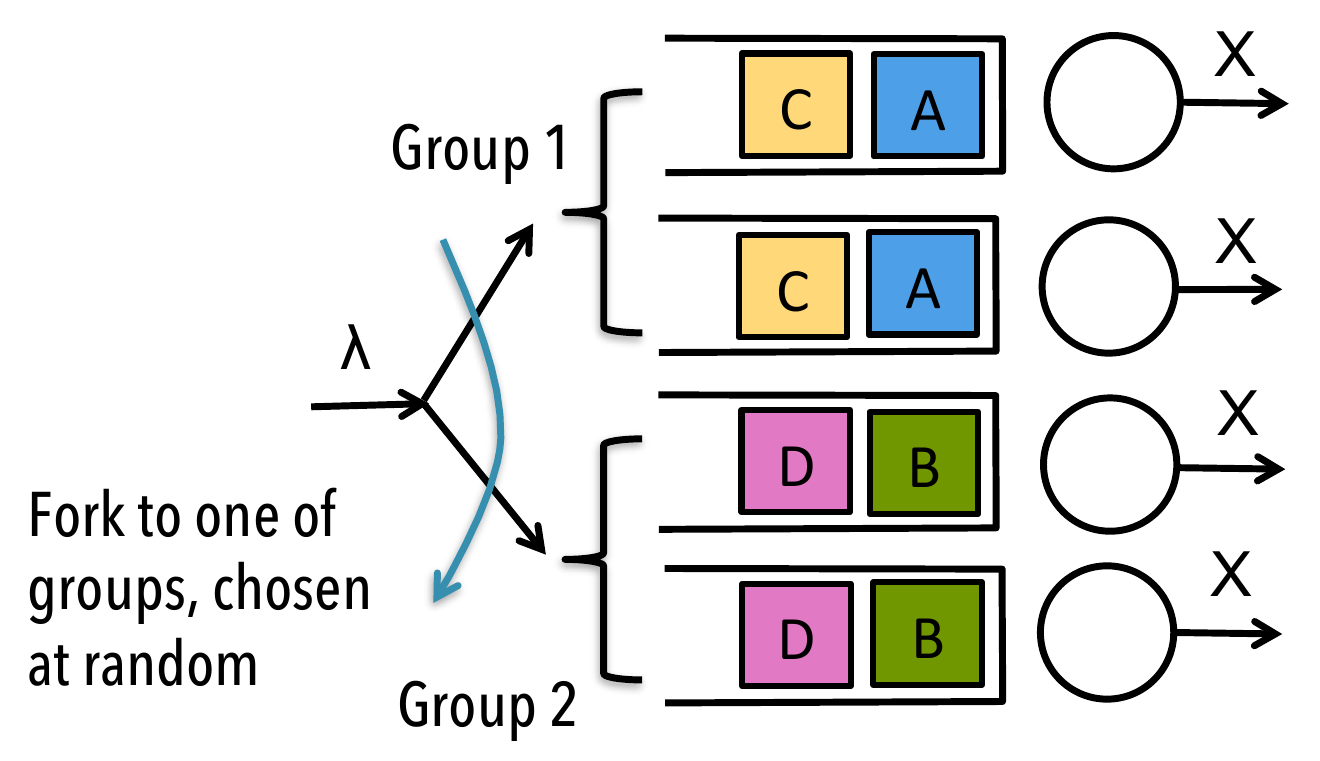}
	\caption{Group-based random}
\end{subfigure}
~
\begin{subfigure}[t]{0.4 \linewidth}
    \centering
   \includegraphics[width=1.55in]{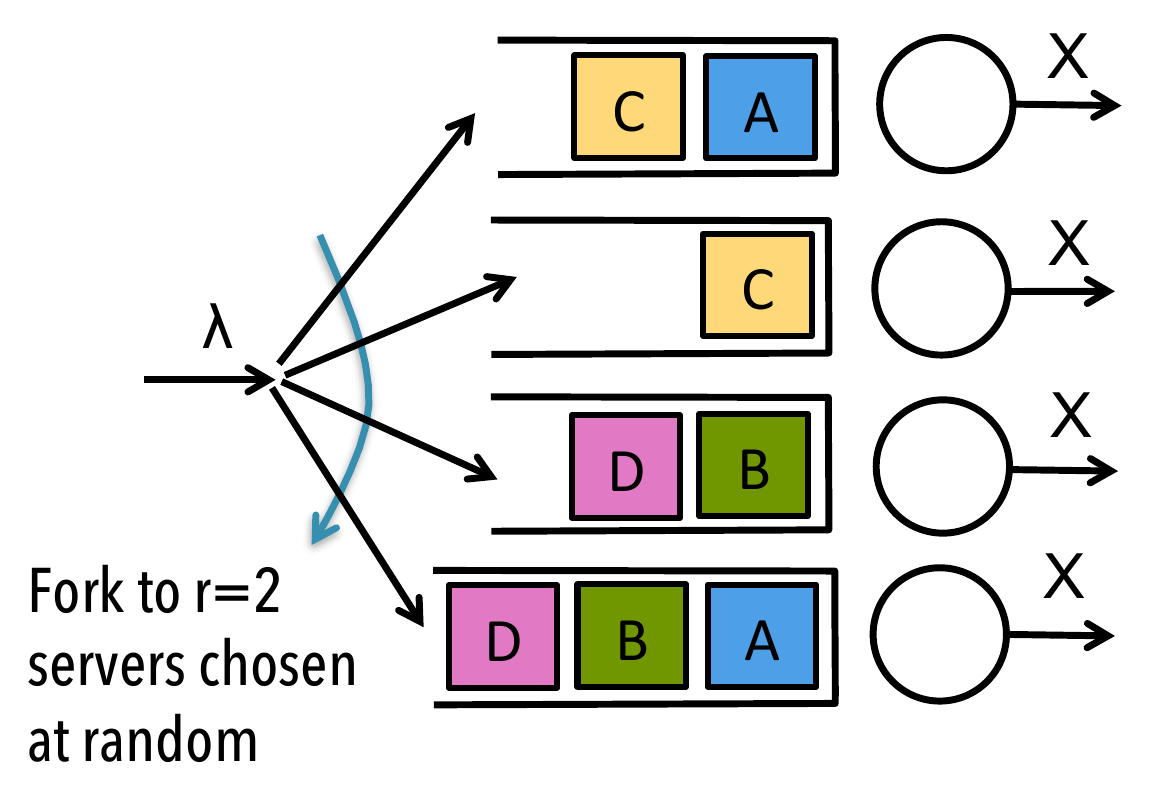}
\caption{Uniform random}
\end{subfigure}
\caption{$(4,2,1)$ partial-fork-join system, where each job is forked to $r=2$ servers, chosen according to the group-based random or uniform random policies. \label{fig:partial_fork}}
\end{figure}

\subsection{Latency-Cost Analysis}
%The forking policy affects the relative starting times of the tasks of a job. Staggered task start times can save the computing cost $\E{C}$, but could cause a loss of diversity resulting in higher latency $\E{T}$. 
%
In the group-based random policy, the job arrivals are split equally across the groups, and each group behaves like an independent $(r,1)$ fork-join system. Thus, the expected latency and cost follow from \Cref{thm:rep_queueing} as given in \Cref{lem:latency_cost_group_based} below.

\begin{lem}[Group-based random]
\label{lem:latency_cost_group_based}
The expected latency and cost when each job is forked to one of $n/r$ groups of $r$ servers each are given by
\begin{align}
\E{T} &= \E{X_{1:r}}  + \frac{\lambda r \E{X_{1:r}^2}}{2(n - \lambda r \E{X_{1:r}})} \label{eqn:E_T_group_based} \\
\E{C} &= r \E{X_{1:r}} \label{eqn:E_C_group_based}
\end{align}
\end{lem}

\begin{proof}
The job arrivals are split equally across the $n/r$ groups, such that the arrival rate to each group is a Poisson process with rate $\lambda r/n$. The $r$ tasks of each job start service at their respective servers simultaneously, and thus each group behaves like an independent $(r,1)$ fork-join system with Poisson arrivals at rate $\lambda r/n$. Hence, the expected latency and cost follow from \Cref{thm:rep_queueing}. 
\end{proof}

Using \eqref{eqn:E_C_group_based} and \Cref{lem:capacity_in_terms_of_EC}, we can infer that the service capacity (maximum supported $\lambda$) for an $(n,r,1)$ system with group-based random policy is
\begin{align}
\lambda_{max} = \frac{n}{r\E{X_{1:r}}} \label{eqn:lmbda_max_group_based}
\end{align}
From \eqref{eqn:lmbda_max_group_based} we can infer that the $r$ that minimizes $r \E{X_{1:r}}$ results in the highest service capacity, and hence the lowest $\E{T}$ in the heavy traffic regime. By \Cref{lem:r_E_X_1_r_trend}, the optimal $r$ is $r=1$ ($r=n$) for log-concave (log-convex) $\bar{F}_X$. For distributions that are neither log-concave nor log-convex, an intermediate $r$ may be optimal and we can determine it using \Cref{lem:latency_cost_group_based}. For example, \Cref{fig:ET_vs_EC_group_based_var_r} shows a plot of latency versus cost as given by \Cref{lem:latency_cost_group_based} for $n=12$ servers. The task service time $X \sim \Pareto(1,2.2)$. Each job is replicated at $r$ servers according to the group-based random policy, with $r$ varying along each curve. Initially increasing $r$ reduces the latency, but beyond $r^*$, the replicas cause an increase in the queueing delay. This increase in queueing delay is more dominant for higher $\lambda$. Thus the optimal  $r^*$ decreases as $\lambda$ increases. 

\begin{figure}[t]
  %  \begin{minipage}[t]{0.48\linewidth}
    \centering
\includegraphics[width=3.5in]{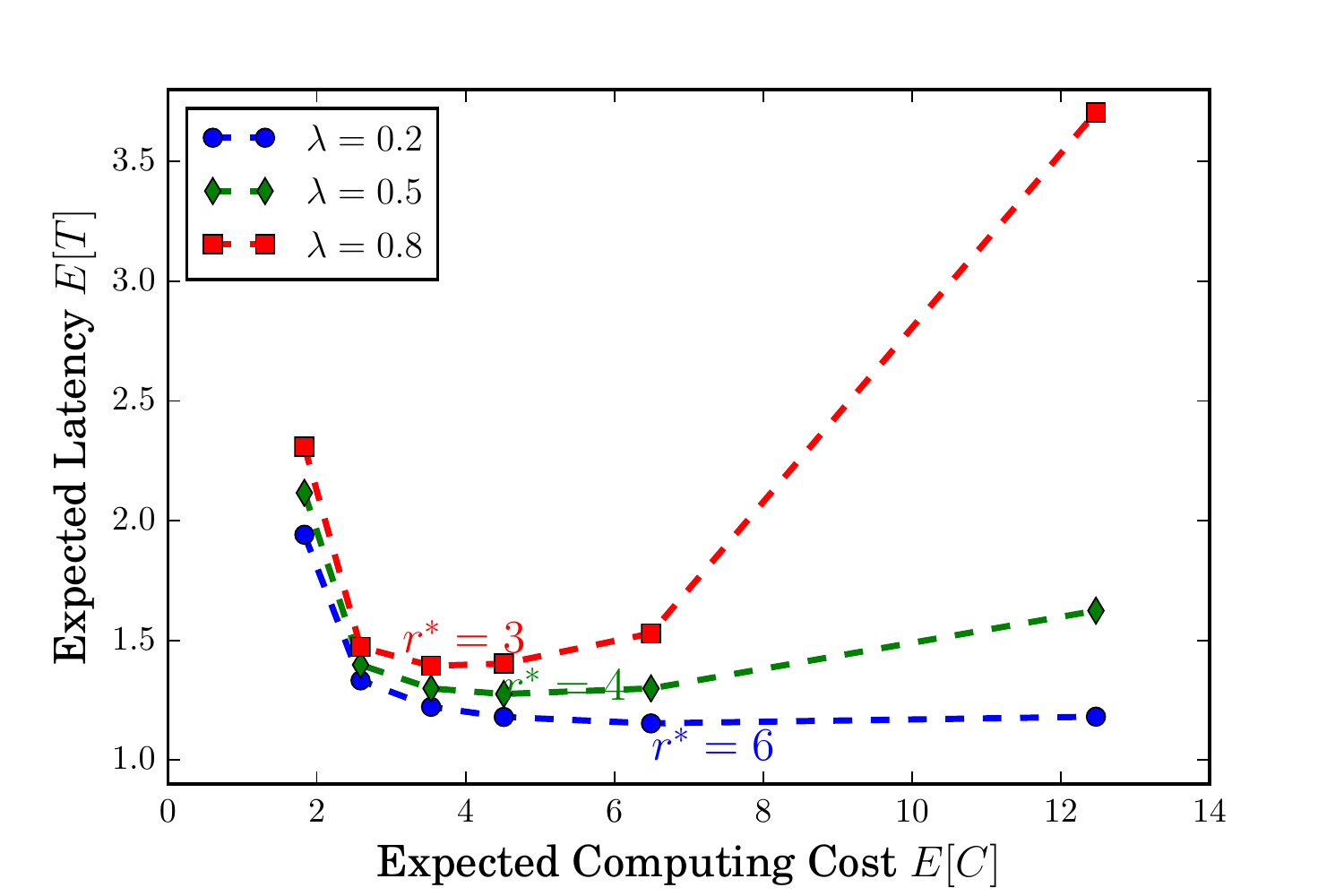}
\caption{Analytical plot of latency versus cost for $n=12$ servers. Each job is replicated at $r$ servers chosen by the group-based random policy, with $r$ increasing as $1$, $2$, $3$, $4$, $6$, and $12$ along each curve. The task service time $X \sim \Pareto(1,2.2)$. As $\lambda$ increases the replicas increase queueing delay. Thus the optimal $r^*$ that minimizes $\E{T}$ shifts downward as $\lambda$ increases. \label{fig:ET_vs_EC_group_based_var_r}}
%    \end{minipage}
%    \hspace{0.04\linewidth}    
\end{figure}

For other symmetric policies, it is difficult to get an exact analysis of $\E{T}$ and $\E{C}$ because the tasks of a job can start at different times. However, we can get bounds on $\E{C}$ depending on the log-concavity of $X$, given in \Cref{thm:E_C_r_trend} below.

\begin{thm}
\label{thm:E_C_r_trend}
Consider an $(n,r,1)$ partial-fork join system, where a job is forked into tasks at $r$ out of $n$ servers chosen according to a symmetric policy. For any relative task start times $t_i$, $\E{C}$ can be bounded as follows.
\begin{align} 
 r \E{X_{1:r}} \geq \E{C} &\geq \E{X} \quad \quad \text{if } \bar{F}_X \text{ is  log-concave} \label{eqn:E_C_bounds_log_concave}\\
\E{X} \geq \E{C} &\geq r \E{X_{1:r}}  \quad \text{if } \bar{F}_X \text{ is  log-convex} \label{eqn:E_C_bounds_log_convex}
\end{align} 
In the extreme case when $r=1$, $\E{C} = \E{X}$, and when $r = n$, $\E{C} = n \E{X_{1:n}}$. 
\end{thm}

To prove \Cref{thm:E_C_r_trend} we take expectation in \eqref{eqn:C_expr}, and show that for log-concave and log-convex $\bar{F}_X$, we get the bounds in \eqref{eqn:E_C_bounds_log_concave} and \eqref{eqn:E_C_bounds_log_convex}, which are independent of the relative task start times $t_i$. 
The detailed proof is given in Appendix~\ref{sec:rep_queueing_proofs}. 

\begin{figure}
%\begin{minipage}[t]{0.48\linewidth}
\centering
\includegraphics[width=3.5in]{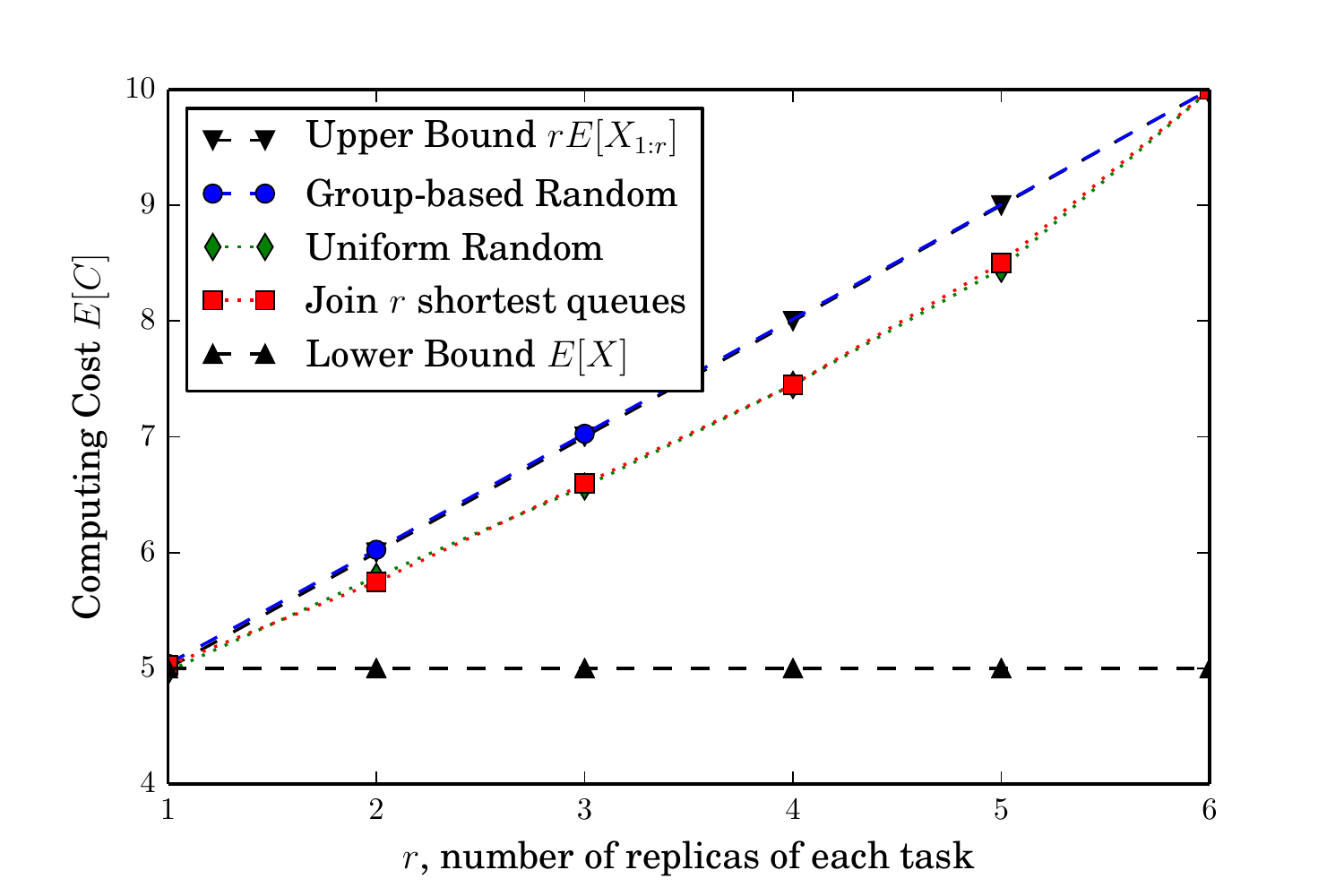}
\caption{Expected cost $\E{C}$ versus $r$ for $X \sim \ShiftedExp(1, 0.25)$, $n=6$ servers, arrival rate $\lambda = 0.5$ and different scheduling policies. The upper bound $r \E{X_{1:r}}$ is exact for the group-based random policy, and fairly tight for other policies.\label{fig:EC_bounds_shifted_exp_vs_r}}
%\end{minipage}
\end{figure}

In \Cref{fig:EC_bounds_shifted_exp_vs_r} we show the bounds given by \eqref{eqn:E_C_bounds_log_concave} for log-concave distributions alongside simulation values, for different scheduling policies. The service time $X \sim \ShiftedExp(1, 0.25)$, and arrival rate $\lambda = 0.5$. Since all replicas start simultaneously with the group-based random policy, the upper bound $\E{C} \geq r \E{X_{1:r}}$ is tight for any $r$. For other scheduling policies, the bound is more loose for the policy that staggers relative start times of replicas to a greater extent. %Thus, the uniform random results in slightly lower $and shortest queue policies indicates that the $\E{C}$ is smaller with the uniform random policy because it has a higher chance of assigning a task to queues of unequal lengths, thus staggering relative start times of replicas. The uniform random policy staggers replicas more than the shortest queues policy, and thus $\E{C}$ is slightly smaller.

%If $X$ is a log-convex distribution, then the bounds are reversed as given by \eqref{$r \E{X_{1:r}}$ becomes a lower bound on $\E{C}$ bounds are the lower bound $r \E{X_{1:r}}$the actual $\E{C}$ is close to its lower bound $r \E{X_{1:r}}$.

\subsection{Optimal value of $r$}
We can use the bounds in \Cref{thm:E_C_r_trend} to gain insights into choosing the best $r$ when $\bar{F}_X$ is log-concave or log-convex. In particular, we study two extreme traffic regimes: low traffic ($\lambda \rightarrow 0$) and heavy traffic ($\lambda \rightarrow \lambda^{*}_{max}$), where $\lambda^{*}_{max}$ is the service capacity of the system introduced in \Cref{defn:serv_capacity}. 

%By \Cref{lem:r_E_X_1_r_trend}, $r \E{X_{1:r}}$ is non-decreasing (non-increasing) with $r$ for log-concave (log-convex) $\bar{F}_X$. By this fact and \Cref{thm:E_C_r_trend}, we get the following corollaries about how $\E{C}$ and $\E{T}$ vary with $r$. %result about which $r$ minimizes $\E{C}$.

\begin{coro}[Expected Cost vs. $r$]
\label{coro:EC_vs_r}
For a system of $n$ servers with symmetric forking of each job to $r$ servers, $r=1$ ($r=n$) minimizes the expected cost $\E{C}$ when $\bar{F}_X$ is log-concave (log-convex).
\end{coro}
The proof follows from \Cref{lem:r_E_X_1_r_trend}, $r \E{X_{1:r}}$ is non-decreasing (non-increasing) with $r$ for log-concave (log-convex) $\bar{F}_X$.
%
%From \Cref{coro:EC_vs_r} and \Cref{clm:capacity_in_terms_of_EC} we get the following result on which $r$ minimizes the service capacity.
%
%\begin{coro}[Service Capacity vs. $r$]
%\label{coro:serv_capacity_r}
%For a system of $n$ servers with symmetric forking of each job to $r$ servers, $r=1$ ($r=n$) gives the highest service capacity when $\bar{F}_X$ is log-concave (log-convex). 
%\end{coro}
%
%Now let us now determine the value of $r$ that minimizes expected latency for log-concave and log-convex $\bar{F}_X$. 
%
%When the arrival rate $\lambda$ is small, $\E{T}$ is dominated by the service time $\E{X_{1:r}}$ which is non-increasing in $r$. This can be observed for the group-based forking policy by taking $\lambda \rightarrow 0$ in \eqref{eqn:E_T_group_based}. Thus we get the following result.

\begin{lem}[Expected Latency vs. $r$]
\label{lem:latency_vs_r}
In the low-traffic regime, forking to all servers ($r =n$) gives the lowest $\E{T}$ for any service time distribution $F_X$. In the heavy traffic regime, $r=1$ ($r=n$) gives lowest $\E{T}$ if $\bar{F}_X$ is log-concave (log-convex).
\end{lem}

\begin{proof}
In the low traffic regime with $\lambda \rightarrow 0$, the waiting time in queue tends to zero. Thus all replicas of a task start service at the same time, irrespective of the scheduling policy. Then the expected latency is $\E{T} = \E{X_{1:r}}$, which decreases with $r$. Thus, $r=n$ gives the lower $\E{T}$ for any service time distribution $F_X$.

By \Cref{coro:high_traffic_comp}, the optimal replication strategy in heavy traffic is the one that minimizes $\E{C}$. For log-convex $\bar{F}_X$, $r=n$ achieves the lower bound $\E{C} = n \E{X_{1:n}}$ in \eqref{eqn:E_C_bounds_log_convex} with equality. Thus, $r=n$ is the optimal strategy in the heavy traffic regime. For log-concave $\bar{F}_X$, $r=1$ achieves the lower bound $\E{C} = \E{X}$ in \eqref{eqn:E_C_bounds_log_concave} with equality. Thus, in heavy traffic, $r=1$ gives lowest $\E{T}$ for log-concave $\bar{F}_X$.
\end{proof}

\begin{figure}[t]
    \begin{minipage}[t]{0.48\linewidth}
    \centering
    \includegraphics[width=3.25in]{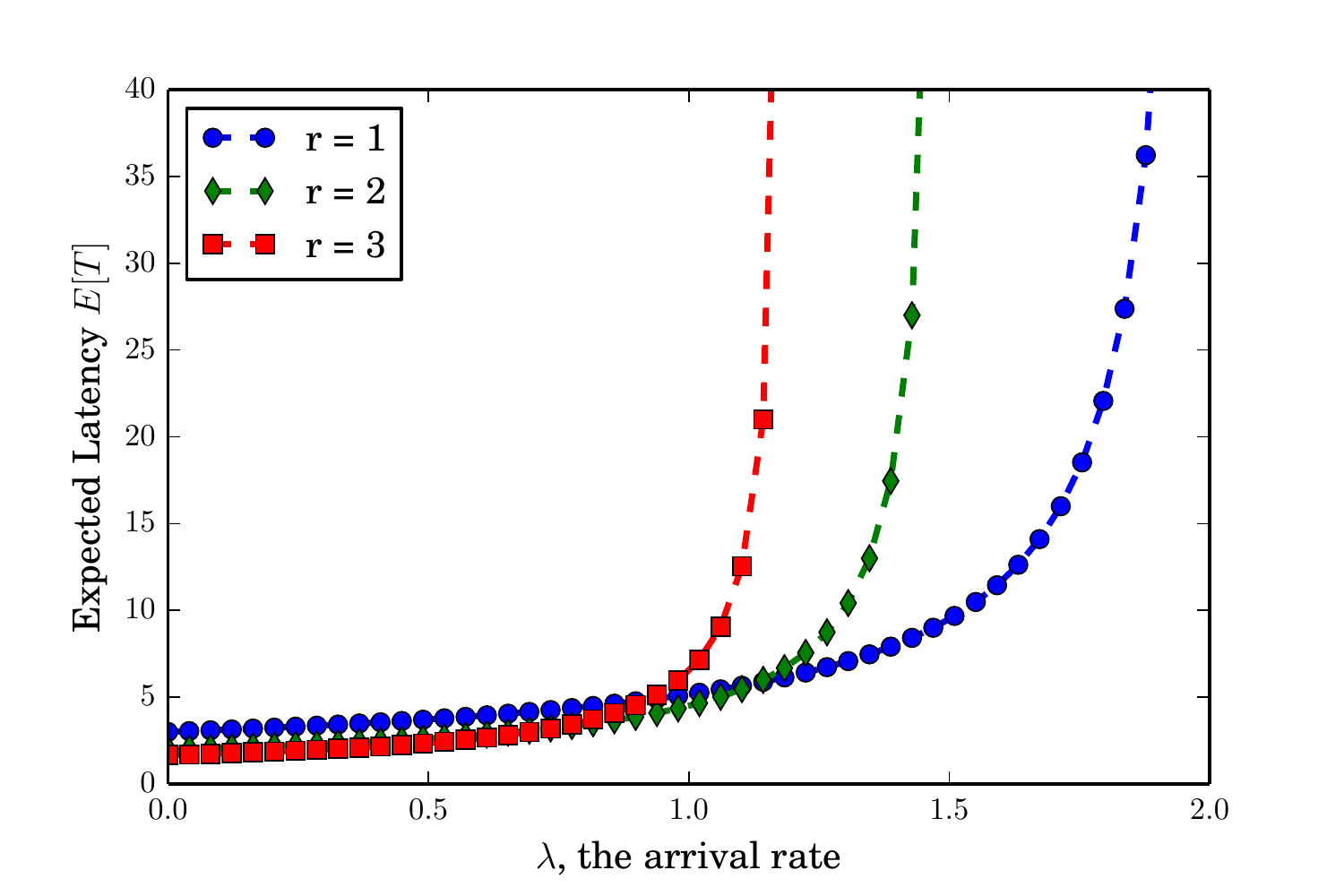}
    \caption{For $X \sim \SExp(1,0.5)$ which is log-concave, forking to less (more) servers reduces expected latency in the low (high) $\lambda$ regime. Each job is replicated at $r$ out of $n=6$ servers, chosen by the group-based random policy. \label{fig:E_T_rep_shifted_exp_vs_lambda_diff_r}}
    \end{minipage}
    \hspace{0.04\linewidth}
    \begin{minipage}[t]{0.48\linewidth}
\centering
\includegraphics[width=3.25in]{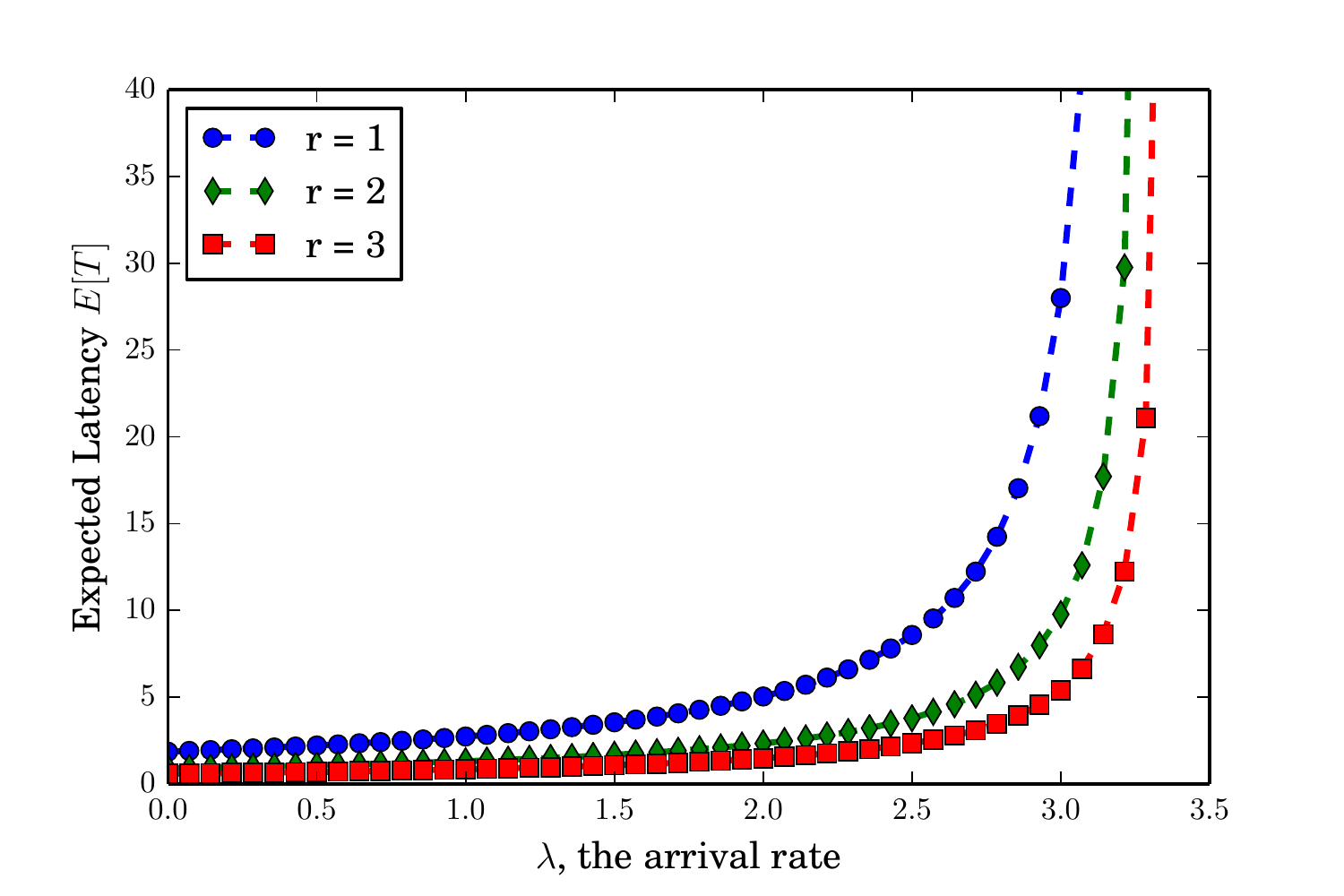}
\caption{For $X \sim \HyperExp(p, \mu_1, \mu_2)$ with $p=0.1$, $\mu_1 = 1.5$, and $\mu_2 = 0.5$ which is log-convex, larger $r$ gives lower expected latency for all $\lambda$. Each job is replicated at $r$ out of $n=6$ servers, chosen according to the group-based random policy. \label{fig:E_T_rep_hyper_exp_vs_lambda_diff_r}}
    \end{minipage}
\end{figure}

\Cref{lem:latency_vs_r} is illustrated by Fig.~\ref{fig:E_T_rep_shifted_exp_vs_lambda_diff_r} and Fig.~\ref{fig:E_T_rep_hyper_exp_vs_lambda_diff_r} where $\E{T}$ calculated analytically using \eqref{eqn:E_T_group_based} is plotted versus $\lambda$ for different values of $r$. Each job is assigned to $r$ servers chosen uniformly at random from $n = 6$ servers. In \Cref{fig:E_T_rep_shifted_exp_vs_lambda_diff_r} the service time distribution is $\SExp(\Delta, \mu)$ (which is log-concave) with $\Delta = 1$ and $\mu = 0.5$. When $\lambda$ is small, more redundancy (higher $r$) gives lower $\E{T}$, but in the high $\lambda$ regime, $r=1$ gives lowest $\E{T}$ and highest service capacity. On the other hand in Fig.~\ref{fig:E_T_rep_hyper_exp_vs_lambda_diff_r}, for a log-convex distribution $\HyperExp(p, \mu_1, \mu_2)$, in the high load regime $\E{T}$ decreases as $r$ increases.

\Cref{lem:latency_vs_r} was previously proven for new-better-than-used (new-worse-than-used) instead of log-concave (log-convex) $\bar{F}_X$ in \cite{shah_when_2013,koole_righter_2008}, using a combinatorial argument. Using \Cref{thm:E_C_r_trend}, we get an alternative, and arguably simpler way to prove this result. Note that our version is weaker because log-concavity implies new-better-than-used but the converse is not true in general (see \Cref{propty:sub_super_additivity} in Appendix~\ref{sec:tail_properties}). 

%Note that the result in \cite{koole_righter_2008} is in fact even stronger because it shows that when $\bar{F}_X$ is new-worse-than-used, $r=n$ is optimal for all $\lambda$ (as observed in Fig.~\ref{fig:E_T_rep_hyper_exp_vs_lambda_diff_r}). 

Due to the network cost of issuing and canceling the replicas, there may be an upper limit $r \leq r_{max}$ on the number of replicas.  The optimal strategy under this constraint is given by \Cref{lem:opt_under_r_max} below.

\begin{lem}[Optimal $r$ under $r \leq r_{max}$]
\label{lem:opt_under_r_max}
For log-convex $\bar{F}_X$, $r=r_{max}$ is optimal. For log-concave $\bar{F}_X$, $r=1$ is optimal in heavy traffic.  
\end{lem}
The proof is similar to \Cref{lem:latency_vs_r} with $n$ replaced by $r_{max}$.

\subsection{Choice of the $r$ servers}
For a given $r$, we now compare different policies of choosing the $r$ servers for each job. The choice of the $r$ servers determines the relative starting times of the tasks. By using the bounds in \Cref{thm:E_C_r_trend} that hold for any relative task start times we get the following result. 

% we know that $\E{C} \geq r \E{X_{1:r}}$ for log-convex $\bar{F}_X$ for a symmetric policy with any relative task start times. Since $\E{C} = r \E{X_{1:r}}$ when all tasks start simutaneously, this implies that staggering the task start times does not help. For log-concave $\bar{F}_X$, since $\E{C} \leq r \E{X_{1:r}}$, higher diversity in the relative starting times of tasks gives lower $\E{T}$. We state this formally is \Cref{coro:optimal_policy} below.

\begin{lem}[Cost of different policies]
\label{lem:cost_diff_policies}
Given $r$, if $\bar{F}_X$ is log-concave (log-convex), the symmetric policy that results in the tasks starting at the same time ($t_i = 0$ for all $1 \leq i \leq r$) results in higher (lower) $\E{C}$ than one that results in $0 < t_i < \infty$ for one or more $i$.
\end{lem}

\begin{proof}
The symmetric policy that results in $t_i = 0$ for all $1 \leq i \leq r$ (for eg.\ the group-based random policy) results in $\E{C} = r \E{X_{1:r}}$. By \Cref{thm:E_C_r_trend}, if $\bar{F}_X$ is log-concave, $\E{C} \leq r \E{X_{1:r}}$ for any symmetric policy. Thus, for log-concave distributions, the symmetric policy that results in $0 < t_i < \infty$ for one or more $i$ gives lower $\E{C}$ than the group-based random policy. On the other hand, for log-convex distributions, $\E{C} \geq r \E{X_{1:r}}$ with any symmetric policy. Thus the policies that result in relative task start times $t_i = 0$ for all $1 \leq i \leq r$ give lower $\E{C}$ than other symmetric policies.
\end{proof}

\begin{lem}[Latency in high $\lambda$ regime]
\label{lem:latency_diff_policies}
Given $r$, if $\bar{F}_X$ is log-concave (log-convex), the symmetric policy that results in the tasks starting at the same time ($t_i = 0$ for all $1 \leq i \leq r$) results in higher (lower) $\E{T}$ in the heavy traffic regime than one that results in $0 < t_i < \infty$ for some $i$.
\end{lem}

\begin{proof}
By \Cref{coro:high_traffic_comp}, the optimal replication strategy in heavy traffic is the one that minimizes $\E{C}$. Then the proof follows from \Cref{lem:cost_diff_policies}.
\end{proof}

\Cref{lem:latency_diff_policies} is illustrated by Fig.~\ref{fig:E_T_shifted_exp_vs_lambda_diff_policies} and Fig.~\ref{fig:E_T_hyper_exp_vs_lambda_diff_policies} for $n=6$ and $r=3$. The simulations are run for $100$ workloads with $1000$ jobs each. The $r$ tasks may start at different times with the uniform random policy, whereas they always start simultaneously with group-based random policy. Thus, in the high $\lambda$ regime, the uniform random policy results in lower latency for log-concave $\bar{F}_X$, as observed in Fig.~\ref{fig:E_T_shifted_exp_vs_lambda_diff_policies}. But for log-convex $\bar{F}_X$, group-based forking is better in the high $\lambda$ regime as seen in Fig.~\ref{fig:E_T_hyper_exp_vs_lambda_diff_policies}. For low $\lambda$, uniform random policy is better for any $\bar{F}_X$ because it gives lower expected waiting time in queue. 

\begin{figure}[t]
    \begin{minipage}[t]{0.48\linewidth}
    \centering
    \includegraphics[width=3.25in]{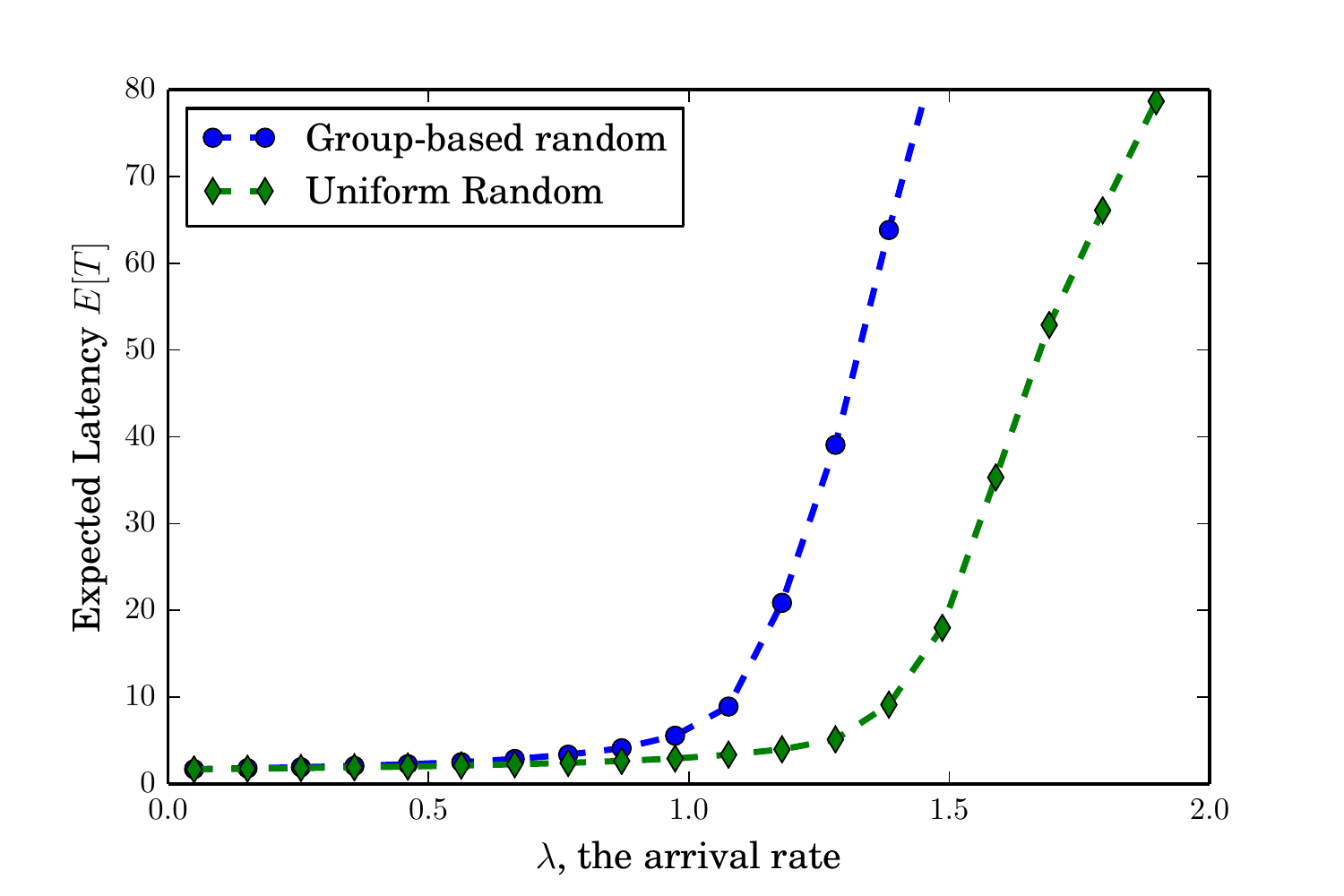}
    \caption{For service time distribution $\SExp(1, 0.5)$ which is log-concave, uniform random scheduling (which staggers relative task start times) gives lower $\E{T}$ than group-based random for all $\lambda$. The system parameters are $n= 6$, $r=3$. \label{fig:E_T_shifted_exp_vs_lambda_diff_policies}}
    \end{minipage}
    \hspace{0.04\linewidth}
    \begin{minipage}[t]{0.48\linewidth}
\centering
\includegraphics[width=3.25in]{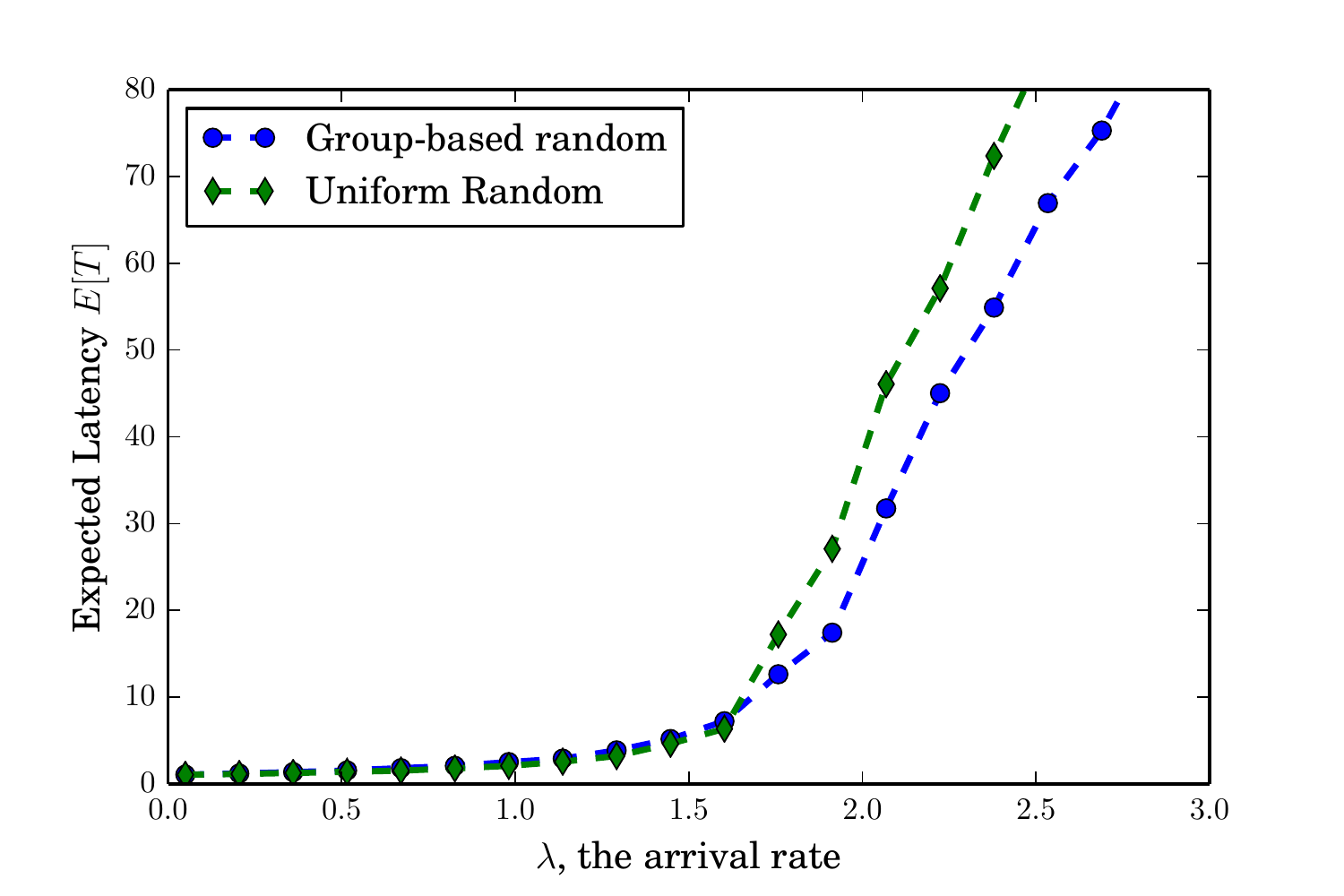}
\caption{ For service time distribution $\HyperExp(0.1, 2.0, 0.2)$ which is log-convex, group-based scheduling gives lower $\E{T}$ than uniform random in the high $\lambda$ regime. The system parameters are $n= 6$, $r=3$. \label{fig:E_T_hyper_exp_vs_lambda_diff_policies}}
    \end{minipage}
\end{figure}

\section{THE GENERAL $k$ CASE}
\label{sec:coded_with_queueing}
We now move to general $k$ case, where a job requires any $k$ out of $n$ tasks to complete. In practice, the general $k$ case arises in large-scale parallel computing frameworks such as MapReduce, and in content download from coded distributed storage systems. In this section we present bounds on the latency and cost of the $(n,k)$ fork-join and $(n,k)$ fork-early-cancel systems. In \Cref{subsec:diversity_parallelism} we demonstrate an interesting diversity-parallelism trade-off in choosing $k$. %Having larger $k$ means smaller size tasks, and thus smaller expected service time $X$ per task. But the diversity in waiting for $k$ out of $n$ reduces as $k$ increases. We demonstrate this trade-off in \Cref{subsec:diversity_parallelism}.

%\TODO{Explain why analysis of general $k$ is hard}

\subsection{Latency and Cost of the $(n,k)$ fork-join system}
Unlike the $k=1$ case, for general $k$ exact analysis is hard because multiple jobs can be in service simultaneously (for e.g.\ Job A and Job B in Fig.~\ref{fig:fork_join_queue}). Even for the $k=n$ case studied in \cite{nelson_tantawi,varki_merc_chen}, only bounds on latency are known. We generalize those latency bounds to any $k$, and also provide bounds on cost $\E{C}$. The analysis of $\E{C}$ can be used to estimate the service capacity using \Cref{lem:capacity_in_terms_of_EC}.

\begin{thm}[Bounds on Latency]
\label{thm:latency_bnds_gen}
The latency $\E{T}$ is bounded as follows.
\begin{align}
\E{T} &\leq \E{X_{k:n}} + \frac{\lambda \E{X_{k:n}^2}}{2(1 - \lambda \E{X_{k:n}})}, \label{eqn:upper_bnd_gen} \\
\E{T} &\geq \E{X_{k:n}} + \frac{\lambda \E{X_{1:n}^2}}{2(1 - \lambda\E{X_{1:n}})}. \label{eqn:lower_bnd_gen}
\end{align}
%where $\E{X_{k:n}}$ and $\V{X_{k:n}}$ are the mean and variance of $X_{k:n}$, the $k^{th}$ smallest in $n$  i.i.d.\ random variables $X_1$, $X_2$, $\cdots$, $X_{\nFork}$, with $X_i \sim F_X$. 
\end{thm}

%\begin{proof}[Proof of \Cref{thm:latency_bnds_gen}]
%To find the upper bound on latency, we consider a related queueing system called the split-merge queueing system. In the split-merge system all the queues are blocked and cannot serve other requests until $k$ chunks are downloaded by the user. Since the caches which finish serving any of the first $k$ tasks are not blocked in the fork-join system, the latency of the split-merge system serves as an upper bound on that of the fork-join system. In the split-merge system we observe that jobs are served one-by-one, and no two jobs are served simultaneously. Thus, the fork-join system reduces to an $M/G/1$ queue with Poisson arrival rate $\lambda n/m$, and service time $X_{k:n}$, the $k^{th}$ order statistic of the i.i.d.\ service times $X_1$, $X_2, \cdots X_n$. The expected latency of an $M/G/1$ queue is given by the Pollaczek-Khinchine formula \cite{dsp_gallager}, and it reduces to the upper bound in \eqref{eqn:upper_bnd_gen}.
%
%To find the lower bound we consider a genie system where the job requires $k$ out of $n$ tasks to complete, but all jobs arriving before it require only $1$ task to finish. Then the service time is $\E{X_{k:n}}$. The expected waiting time in queue is equal to the waiting time in \eqref{eqn:upper_bnd_gen} with $k$ set to $1$.  Adding the expected service and the lower bound on expected waiting time, we get the lower bound \eqref{eqn:lower_bnd_gen} on the expected latency.
%\end{proof}

The proof is given in Appendix~\ref{sec:coded_queueing_proofs}. In Fig.~\ref{fig:E_T_pareto_vs_k} 
%and Fig.~\ref{fig:E_T_shifted_exp_vs_k} 
we plot the bounds on latency alongside the simulation values for Pareto service time. The upper bound \eqref{eqn:upper_bnd_gen} becomes more loose as $k$ increases, because the split-merge system considered to get the upper bound (see proof of \Cref{thm:latency_bnds_gen}) becomes worse as compared to the fork-join system. For the special case $k=n$ we can improve the upper bound in \Cref{lem:tighter_upper_bnd} below, by generalizing the approach used in \cite{nelson_tantawi}.

\begin{lem}[Tighter Upper bound when $k=n$]
\label{lem:tighter_upper_bnd}
For the case $k=n$, another upper bound on latency is given by,
\begin{align}
\E{T} \leq \E{ \max \left( R_1, R_2, \cdots R_n \right) },
\end{align}
where $R_i$ are i.i.d.\ realizations of the response time $R$ of an $M/G/1$ queue with arrival rate $\lambda$, service time distribution $F_X$. 
\end{lem}

The proof is given in Appendix~\ref{sec:coded_queueing_proofs}. Transform analysis \cite[Chapter~25]{mor_book} can be used to determine the distribution of $R$, the response time of an $M/G/1$ queue in terms of $F_X(x)$. The Laplace-Stieltjes transform $R(s)$ of the probability density function of $f_R(r)$ of $R$ is given by,
\begin{align}
R(s) &= \frac{ s X(s) \left(1- \frac{\lambda}{\E{X}} \right)}{s - \lambda (1-X(s))}\label{eqn:R_i_pdf_laplace_tx},
\end{align}
where $X(s)$ is the Laplace-Stieltjes transform of the service time distribution $f_X(x)$. 

%\begin{proof}[Proof of \Cref{lem:tighter_upper_bnd}]
%The bound above is a generalization of the bound for the $(n,n)$ fork-join system with exponential service time presented in \cite{nelson_tantawi}. To find the bound, we first observe that the response times of the $n$ queues form a set of associated random variables \cite{assoc_rand_vars}. Then we use the property of associated random variables that their expected maximum is less than that for independent variables with the same marginal distributions. Unfortunately, this approach cannot be directly extended to the $(n,k)$ fork-join system with $k < n$ because this property of associated variables does not hold for the $k^{th}$ order statistic for $k<n$.
%\end{proof}

The lower bound on latency \eqref{eqn:lower_bnd_gen} can be improved for shifted exponential $F_X$, generalizing the approach in \cite{varki_merc_chen} based on the memoryless property of the exponential tail. %We omit the result here for the purpose of brevity. 

\begin{figure}[th]
    \begin{minipage}[t]{0.48\linewidth}
 	\centering
	\includegraphics[width=3.5in]{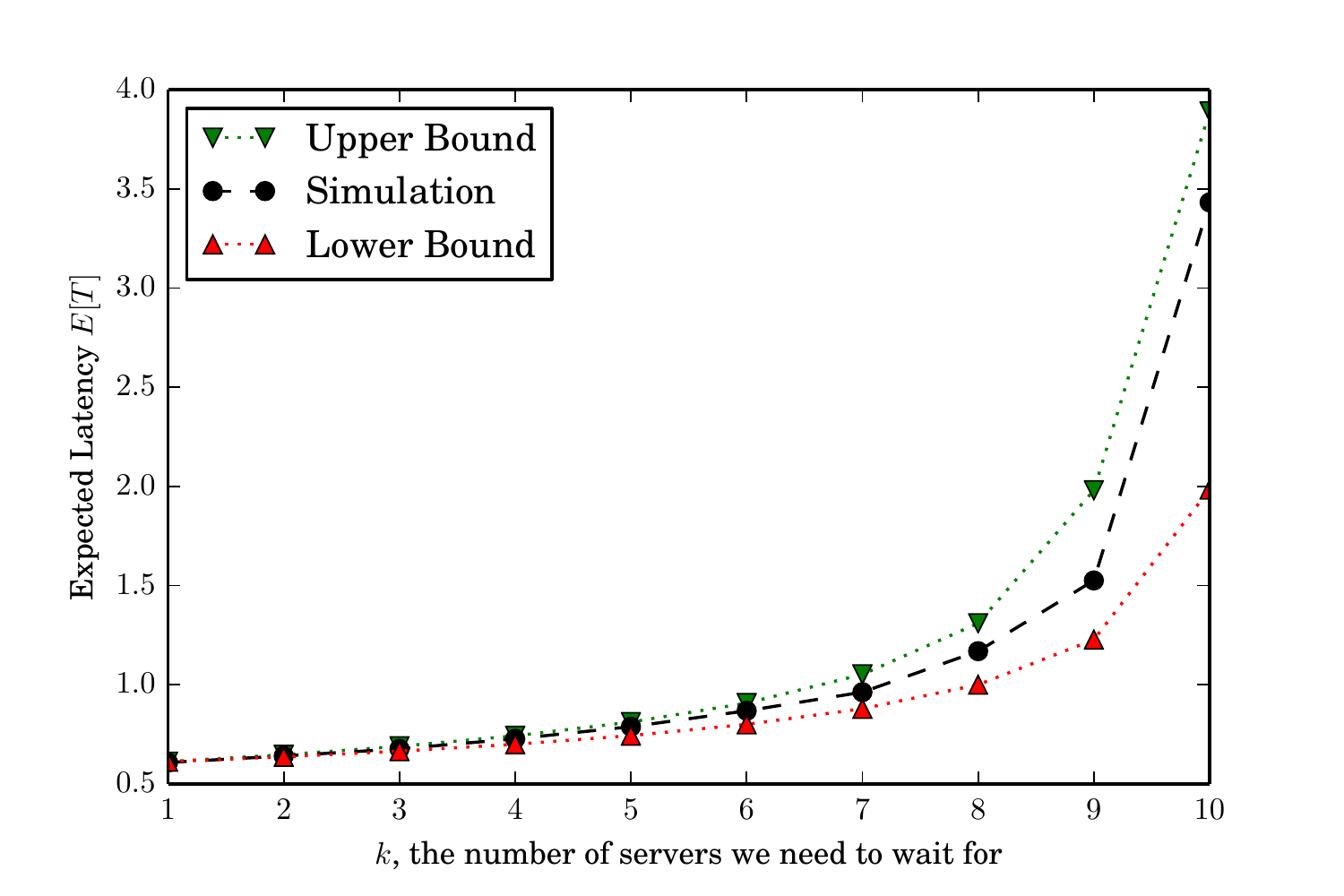}
	\caption{Bounds on latency $\E{T}$ versus $k$ (\Cref{thm:latency_bnds_gen}), alongside simulation values. The service time $X\sim \Pareto(0.5,2.5)$, $n=10$, and $\lambda =0.5$. A tigher upper bound for  $k=n$ is evaluated using \Cref{lem:tighter_upper_bnd}.\label{fig:E_T_pareto_vs_k}}
    \end{minipage}
    \hspace{0.04\linewidth}
    \begin{minipage}[t]{0.48\linewidth}
	\centering
	\includegraphics[width=3.5in]{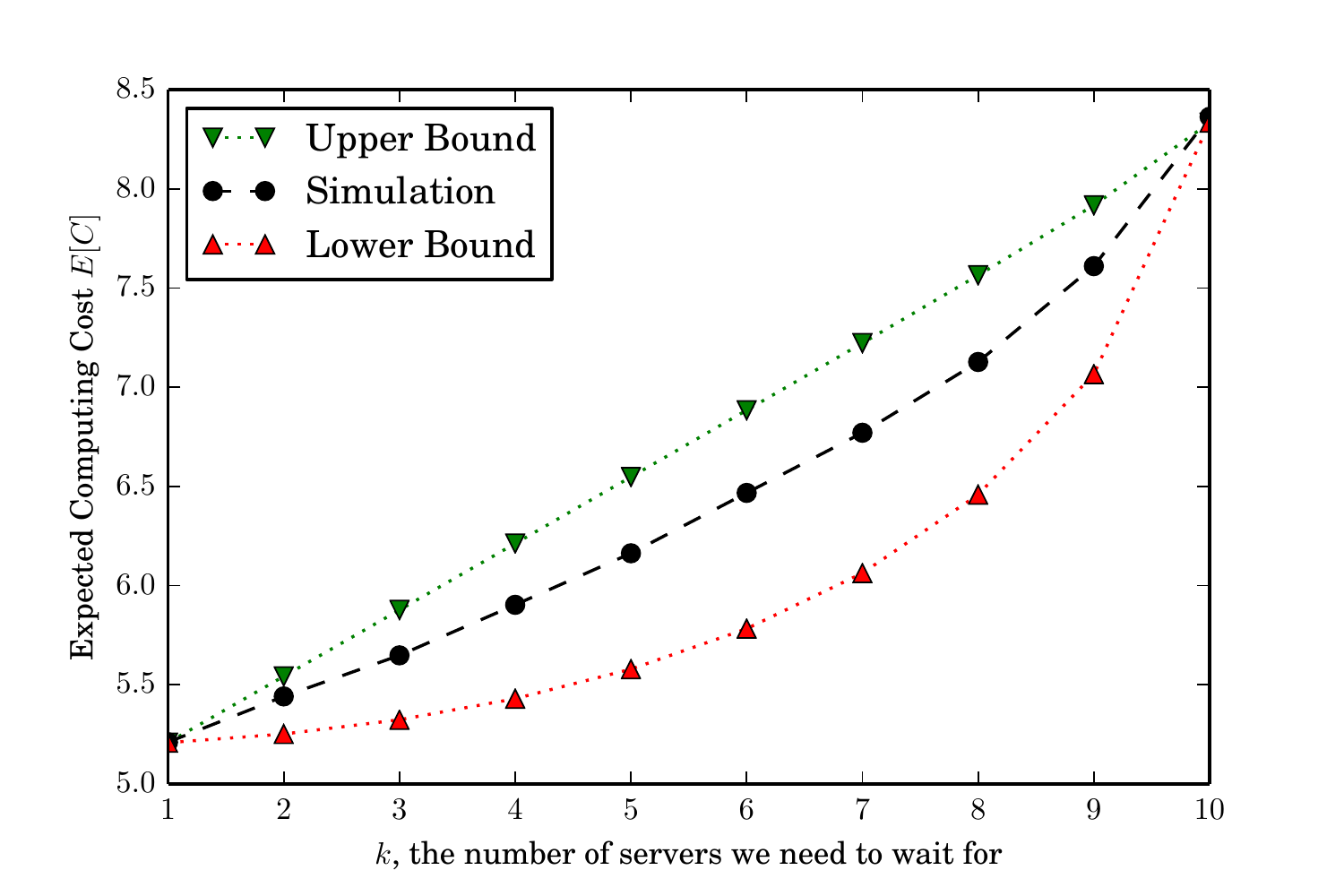}
	\caption{Bounds on cost $\E{C}$ versus $k$ (\Cref{thm:cost_bounds_gen}) alongside simulation values. The service time $X \sim \Pareto(0.5,2.5)$, $n=10$, and $\lambda =0.5$. The bounds are tight for $k=1$ and $k=n$.\label{fig:E_C_pareto_vs_k}}
    \end{minipage}
\end{figure}

\begin{thm}[Bounds on Cost]
\label{thm:cost_bounds_gen}
The expected computing cost $\E{C}$ can be bounded as follows.
\begin{align}
\E{C} &\leq (k-1)\E{X} + (n-k+1) \E{X_{1:n-k+1}} \label{eqn:E_C_gen_upper_bnd} \\
\E{C} &\geq \sum_{i=1}^{k} \E{X_{i:n}} + (n-k) \E{X_{1:n-k+1}} \label{eqn:E_C_gen_lower_bnd}
\end{align}
\end{thm}

The proof is given in Appendix~\ref{sec:coded_queueing_proofs}. Fig.~\ref{fig:E_C_pareto_vs_k} shows the bounds alongside the simulation plot of the computing cost $\E{C}$ when $F_X$ is $\Pareto(x_m, \alpha)$ with $x_m = 0.5$ and $\alpha = 2.5$. The arrival rate $\lambda = 0.5$, and $n=10$ with $k$ varying from $1$ to $10$ on the x-axis. The simulation is run for $100$ iterations of $1000$ jobs. We observe that the bounds on $\E{C}$ are tight for $k=1$ and $k=n$, which can also be inferred from \eqref{eqn:E_C_gen_upper_bnd} and \eqref{eqn:E_C_gen_lower_bnd}.

\subsection{Diversity-Parallelism Trade-off}
\label{subsec:diversity_parallelism}

\begin{figure}[t]
    \begin{minipage}[t]{0.48\linewidth}
\centering
\includegraphics[width=3.5in]{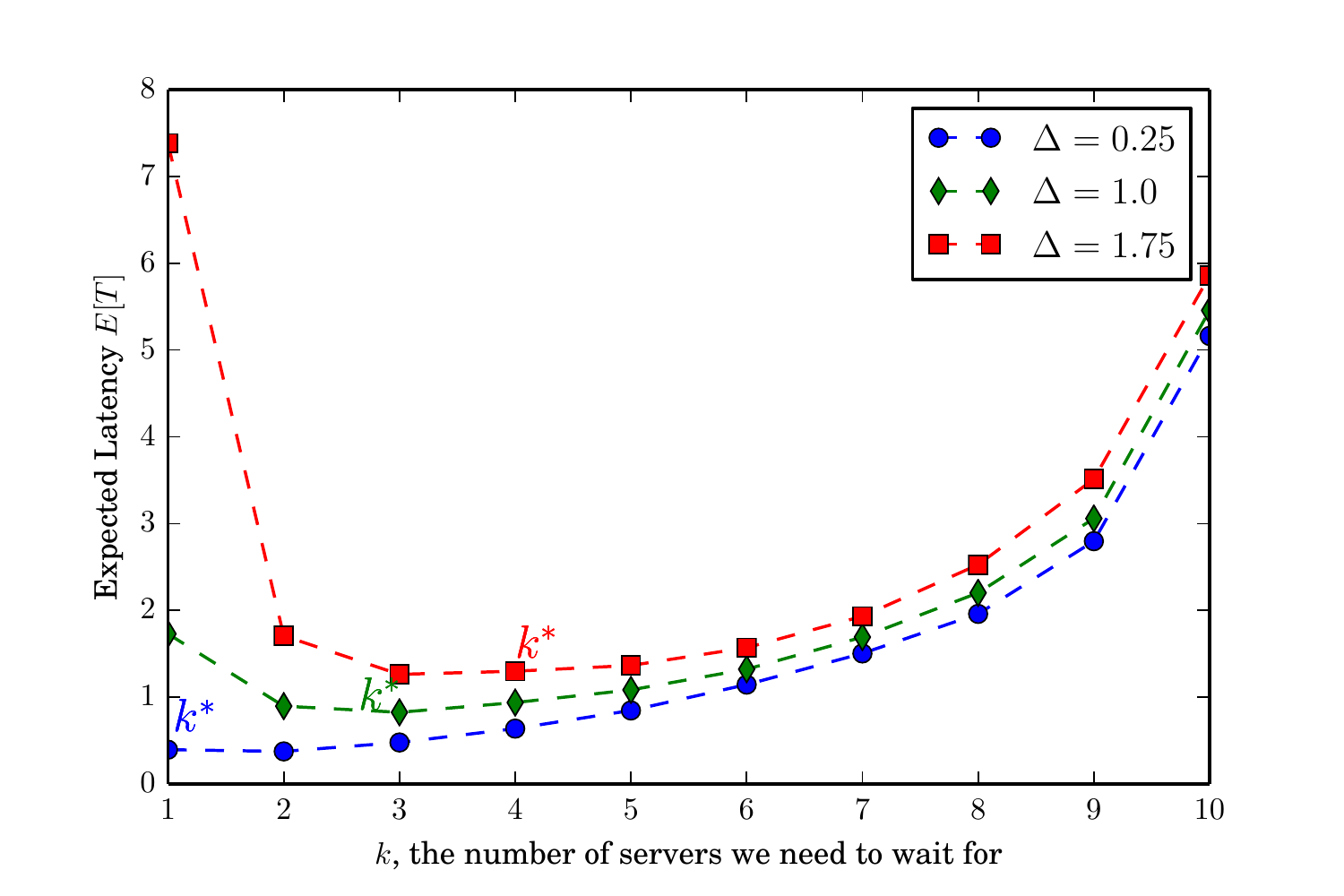}
\caption{Expected latency versus $k$ for task service time $X \sim \SExp(\Delta/k , 1.0)$, and arrival rate $\lambda =0.5$. As $k$ increases, we lose diversity but the parallelism benefit is higher because each task is smaller.  \label{fig:ET_shifted_exp_vs_k}}
    \end{minipage}
    \hspace{0.04\linewidth}
    \begin{minipage}[t]{0.48\linewidth}
\centering
\includegraphics[width=3.5in]{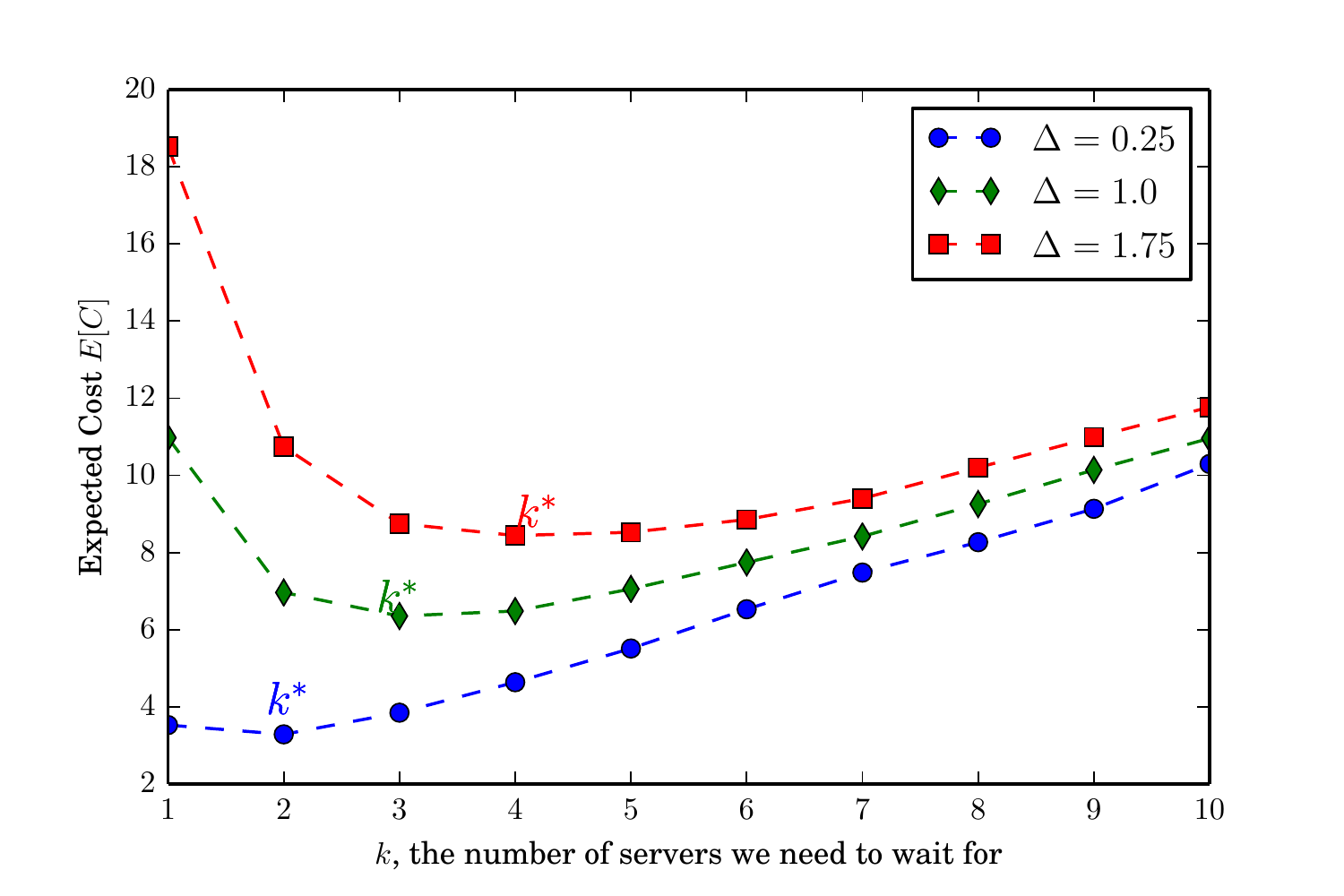}
\caption{Expected cost versus $k$ for task service time $X \sim \SExp(\Delta/k , 1.0)$, and arrival rate $\lambda =0.5$. As $k$ increases, we lose diversity but the parallelism benefit is higher because each task is smaller.   \label{fig:EC_shifted_exp_vs_k}}
    \end{minipage}
\end{figure}

%\begin{figure}[t]
%\centering
%\includegraphics[width=3.5in]{E_T_shifted_exp_vs_k.pdf}
%\caption{Expected latency versus $k$ for task service time $X \sim \SExp(\Delta/k , 1.0)$, and arrival rate $\lambda =0.5$. As $k$ increases, we lose diversity but the parallelism benefit is higher because each task is smaller.  \label{fig:ET_shifted_exp_vs_k}}
%\end{figure}
%
%\begin{figure}[t]
%\centering
%\includegraphics[width=3.5in]{E_C_shifted_exp_vs_k.pdf}
%\caption{Expected cost versus $k$ for task service time $X \sim \SExp(\Delta/k , 1.0)$, and arrival rate $\lambda =0.5$. As $k$ increases, we lose diversity but the parallelism benefit is higher because each task is smaller.   \label{fig:ET_shifted_exp_vs_k}}
%\end{figure}

%
%\begin{figure}[t]
%\centering
%\includegraphics[width=3.5in]{ET_vs_EC_shifted_exp_var_n.pdf}
%\caption{Latency $\E{T}$ versus computing cost $\E{C}$. The service time is $\SExp(\Delta ,\mu)$, with $\mu = 1.0$, and $\Delta = 1/k$, proportional to the size of each task. The arrival rate $\lambda =0.5$. \label{fig:ET_vs_EC_shifted_exp_var_n}}
%\end{figure}

In Fig.~\ref{fig:E_T_pareto_vs_k} we observed the expected latency increases with $k$, because we need to wait for more tasks to complete, and the service time $X$ is independent of $k$. But in most computing and storage applications, the service time $X$ decreases as $k$ increases because each task becomes smaller. We refer to this as the `parallelism benefit' of splitting a job into more tasks. But as $k$ increases, we lose the `diversity benefit' provided by redundant tasks and having to wait only for a subset of the tasks to finish. Thus, there is a diversity-parallelism trade-off in choosing the optimal $k^*$ that minimizes latency $\E{T}$. We demonstrate the diversity-parallelism trade-off in simulation plot Fig.~\ref{fig:ET_shifted_exp_vs_k} for service time $X \sim \SExp(\Delta_k ,\mu)$, with $\mu = 1.0$, and $\Delta_k = \Delta/k$. As $k$ increases, we lose diversity but the parallelism benefit is higher because each task is smaller. As $\Delta$ increases, the optimal $k^*$ shifted upward because the service time distribution becomes `less random' and so there is less diversity benefit.

We can also observe the diversity-parallelism trade-off mathematically in the low traffic regime, for $X \sim \SExp(\Delta/k, \mu)$. If we take $\lambda \rightarrow 0$ in \eqref{eqn:lower_bnd_gen} and \eqref{eqn:upper_bnd_gen}, both bounds coincide and we get,
\begin{align}
\lim_{\lambda \rightarrow \infty} \E{T} &= \E{X_{k:n}} % \label{eqn:diversity_parallelism_1}\\
= \frac{\Delta}{k} + \frac{H_n - H_{n-k}}{\mu} \label{eqn:diversity_parallelism_2},
\end{align}
where $H_n = \sum_{i=1}^{n} 1/i $, the $n^{th}$ harmonic number. The parallelism benefit comes from the first term in \eqref{eqn:diversity_parallelism_2}, which reduces with $k$. The diversity of waiting for $k$ out of $n$ tasks causes the second term to increase with $k$. The optimal $k^*$ that minimizes \eqref{eqn:diversity_parallelism_2} strikes a balance between these two opposing trends.

Fig.~\ref{fig:EC_shifted_exp_vs_k} shows a similar diversity-parallelism trade-off in choosing $k$ to minimize the computing cost $\E{C}$. In the heavy traffic regime, by \Cref{coro:high_traffic_comp} the policy that minimizes $\E{C}$ also minimizes $\E{T}$. Thus the same $k^*$ will minimize both $\E{T}$ and $\E{C}$.

%In Fig.~\ref{fig:ET_vs_EC_shifted_exp_var_n} we plot the latency versus cost when the task service time $X$ is shifted exponential $\SExp(\Delta_k ,\mu)$, with $\mu = 0.75$, and $\Delta_k = 1/k$.  For a given cost budget $\E{C} < \gamma$,  the optimal $k$ and $n$ that gives minimum latency changes with $\gamma$. Smaller $k$ and $n$ works better on a tighter cost budget.

\subsection{Latency and Cost of the $(n,k)$ fork-early-cancel system}
We now analyze the latency and cost of the $(n,k)$ fork-early-cancel system where the redundant tasks are canceled as soon as any $k$ tasks start service. 
%
%\begin{figure}[t]
%\centering
%\includegraphics[width=3.5in]{E_T_normal_and_early_vs_Delta.pdf}
%\caption{The fork-join system parameters are $n=10$, $k=3$. The service time distribution is a shifted exponential with $\mu = 0.25$. As $\Delta$ increases the diversity advantage of having redundant tasks reduces. As a result, early cancellation of the redundant tasks is more effective in reducing the latency than waiting until $k$ tasks finish.\label{fig:E_T_normal_and_early_vs_Delta}}
%\end{figure}
%
%\begin{figure}[t]
%\centering
%\includegraphics[width=3.5in]{E_C_normal_and_early_vs_Delta.pdf}
%\caption{The fork-join system parameters are $m=20$, $n=10$, $k=3$. The service time distribution is a shifted exponential with $\mu = 0.25$. Early cancellation gives lower server cost than the normal fork-join system, and the cost difference increases with $\Delta$.\label{fig:E_C_normal_and_early_vs_Delta}}
%\end{figure}

\begin{thm}[Latency-Cost with Early Cancellation]
\label{thm:early_cancel_gen}
The cost $\E{C}$ and an upper bound on the expected latency $\E{T}$ with early cancellation is given by
\begin{align}
\E{C} &= k \E{X} \label{eqn:EC_early_cancel_bnd}\\
\E{T} &\leq \E{\max \left( R_1, R_2, \cdots R_k \right)} \label{eqn:ET_early_cancel_bnd}
\end{align}
where $R_i$ are i.i.d.\ realizations of $R$, the reponse time of an $M/G/1$ queue with arrival rate $\lambda k/n$ and service time distribution $F_X$. 
\end{thm}
The proof is given in Appendix~\ref{sec:coded_queueing_proofs}. The Laplace-Stieltjes transform of the response time $R$ of an $M/G/1$ queue with service time distribution $F_X(x)$ and arrival rate is same as \eqref{eqn:R_i_pdf_laplace_tx}, with $\lambda$ replaced by $\lambda k/n$.

By comparing the cost $\E{C} = k \E{X}$ in \eqref{eqn:EC_early_cancel_bnd} to the bounds in Theorem~\ref{thm:cost_bounds_gen} without early cancellation, we can get insights into when early cancellation is effective for a given service time distribution $F_X$. For example, when $\bar{F}_X$ is log-convex, the upper bound in \eqref{eqn:E_C_gen_upper_bnd} is smaller than $k \E{X}$. Thus we can infer that the $(n,k)$ fork-early-cancel system is always worse than the $(n,k)$ fork-join system when $X$ is log-convex. We also observed this phenomenon in \Cref{fig:normal_early_vs_lambda_log_convex} for the $k=1$ case.

\section{GENERAL REDUNDANCY STRATEGY}
\label{sec:heuristic_algo}
From the analysis in \Cref{sec:rep_with_queueing} and \Cref{sec:partial_fork}, we get insights into designing the best redundancy strategy for log-concave and log-convex service time. But it is not obvious to infer the best strategy for arbitrary service time distributions, or when only empirical traces of the service time are given. We now propose such a redundancy strategy to minimize the latency, subject to computing and network cost constraints. This strategy can also be used on traces of task service time when closed-form expressions of $F_X$ and its order statistics are not known.

\subsection{Generalized Fork-join Model}
We first introduce a general fork-join variant that is a combination of the partial fork introduced in \Cref{sec:prob_setup}, and partial early cancellation of redundant tasks.

\begin{defn}[$(n,r_{f}, r, k)$ fork-join system]
For a system of $n$ servers and a job that requires $k$ tasks to complete, we do the following:
\begin{itemize}
\item Fork the job to $r_{f}$ out of the $n$ servers chosen uniformly at random. 
\item When any $r \leq r_{f}$ tasks are at the head of queues or in service already, cancel all other tasks immediately. If more than $r$ tasks start service simultaneously, retain $r$ randomly chosen ones out of them.
\item When any $k \leq r$ tasks finish, cancel all remaining tasks immediately.
\end{itemize}
Note $k$ tasks may finish before some $r$ start service, and thus we may not need to perform the partial early cancellation in the second step above.
\end{defn}

Recall that the $n$ servers have service time distribution $X$ that is i.i.d.\ across the servers and tasks. The $r_{f} - r$ tasks that are canceled early, help find the shortest $r$ out of the $r_{f}$ queues, thus reducing waiting time. From the $r$ tasks retained, waiting for any $k$ to finish provides diversity and hence reduces service time. 

The special cases $(n, n, n,k)$, $(n,n,k,k)$ and  $(n,r,r,k)$ correspond to the $(n,k)$ fork-join and $(n,k)$ fork-early-cancel and $(n,r,k)$ partial-fork-join systems respectively, which are defined in Section~\ref{sec:prob_setup}.

\subsection{Choosing Parameters $r_f$ and $r$}
We propose a strategy to choose $r_{f}$ and $r$ to minimize expected latency $\E{T}$, subject to a computing cost constraint is $\E{C} \leq \gamma$, and a network cost constraint is $r_{f} \leq r_{max}$. We impose the second constraint because forking to more servers results in higher network cost of remote-procedure-calls (RPCs) to launch and cancel the tasks. 

\begin{defn}[Proposed Redundancy Strategy]
\label{defn:heuristic_strategy}
Choose $r_{f}$ and $r$ to minimize $\E{T}$ subject to constraints $\E{C} \leq \gamma$ and $r_{f} \leq r_{max}$. The solutions are
\begin{align}
r^*_{f} =&  r_{max} \label{eqn:heuristic_opt_r_fork}, \\
r^* =& \argmin_{r \in [0, r_{max}]} \hat{T}(r), \quad   s. t.  \quad \hat{C}(r) \leq \gamma \label{eqn:heuristic_opt_r_keep}
\end{align}
where $\hat{T}(r)$ and $\hat{C}(r)$ are estimates of the expected latency $\E{T}$ and cost $\E{C}$, defined as follows:
\begin{align}
\hat{T}(r) &\triangleq \E{X_{k:r}} + \frac{ \lambda r \E{X_{k:r}^2}}{2 (n- \lambda r \E{X_{k:r}})}, \\
\hat{C}(r) &\triangleq r \E{X_{k:r}}.
\end{align}
\end{defn}

To justify the strategy above, observe that for a given $r$, increasing $r_f$ gives higher diversity in finding the queues with the least-work-left and thus reduces latency. Since $r_f - r$ tasks are canceled early before starting service, $r_{f}$ affects $\E{C}$ only mildly, through the relative task start times of $r$ tasks that are retained. So we conjecture that it is optimal to set $r_f = r_{max}$ in \eqref{eqn:heuristic_opt_r_fork}, the maximum value possible under network cost constraints. Changing $r$ on the other hand does affect both the computing cost and latency significantly. Thus to determine the optimal $r$, we minimize $\hat{T}(r)$ subject to constraints $\hat{C}(r) \leq \gamma$ and $r \leq r_{max}$ as given in \eqref{eqn:heuristic_opt_r_keep}. 

The estimates $\hat{T}(r)$ and $\hat{C}(r)$ are obtained by generalizing \Cref{lem:latency_cost_group_based} for group-based random forking to any $k$, and $r$ that may not divide $n$. When the order statistics of $F_X$ are hard to compute, or $F_X$ itself is not explicitly known, $\hat{T}(r)$ and $\hat{C}(r)$ can be also be found using empirical traces of $X$. 

The sources of inaccuracy in the estimates $\hat{T}(r)$ and $\hat{C}(r)$ are as follows.
\begin{enumerate}

\item For $k > 1$, the latency estimate $\hat{T}(r)$ is a generalization of the split-merge queueing upper bound in \Cref{thm:latency_bnds_gen}. Since the bound becomes loose as $k$ increases, the error $\lvert \hat{T}(r)- \E{T}\rvert$ increases with $k$.

\item The estimates $\hat{T}(r)$ and $\hat{C}(r)$ are by definition independent of $r_{f}$, which is not true in practice. As explained above, for $r_f > r$, the actual $\E{T}$ is generally less than $\hat{T}(r)$, and $\E{C}$ can be slightly higher or lower than $\hat{C}(r)$.

\item Since the estimates $\hat{T}(r)$ and $\hat{C}(r)$ are based on group-based forking, they consider that all $r$ tasks start simultaneously. Variability in relative task start times can result in actual latency and cost that are different from the estimates. For example, from \Cref{thm:E_C_r_trend} we can infer that when $\bar{F}_X$ is log-concave (log-convex), the actual computing cost $\E{C}$ is less than (greater than) $\hat{C}(r)$. 
\end{enumerate}

The factor (1) above is the largest source of inaccuracy, especially for larger $k$ and $\lambda$. Since the estimate $\hat{T}_r$ is an upper bound on the actual latency, the $r^*$ and $r_f^*$ recommended by the strategy are smaller than or equal to their optimal values. Factors (2) and (3) only affect the relative task start times and generally result in a smaller error in estimating $\E{T}$ and $\E{C}$.

\subsection{Simulation Results}

\begin{figure}[t]
    \begin{minipage}[t]{0.48\linewidth}
	\centering
	\includegraphics[width=3.5in]{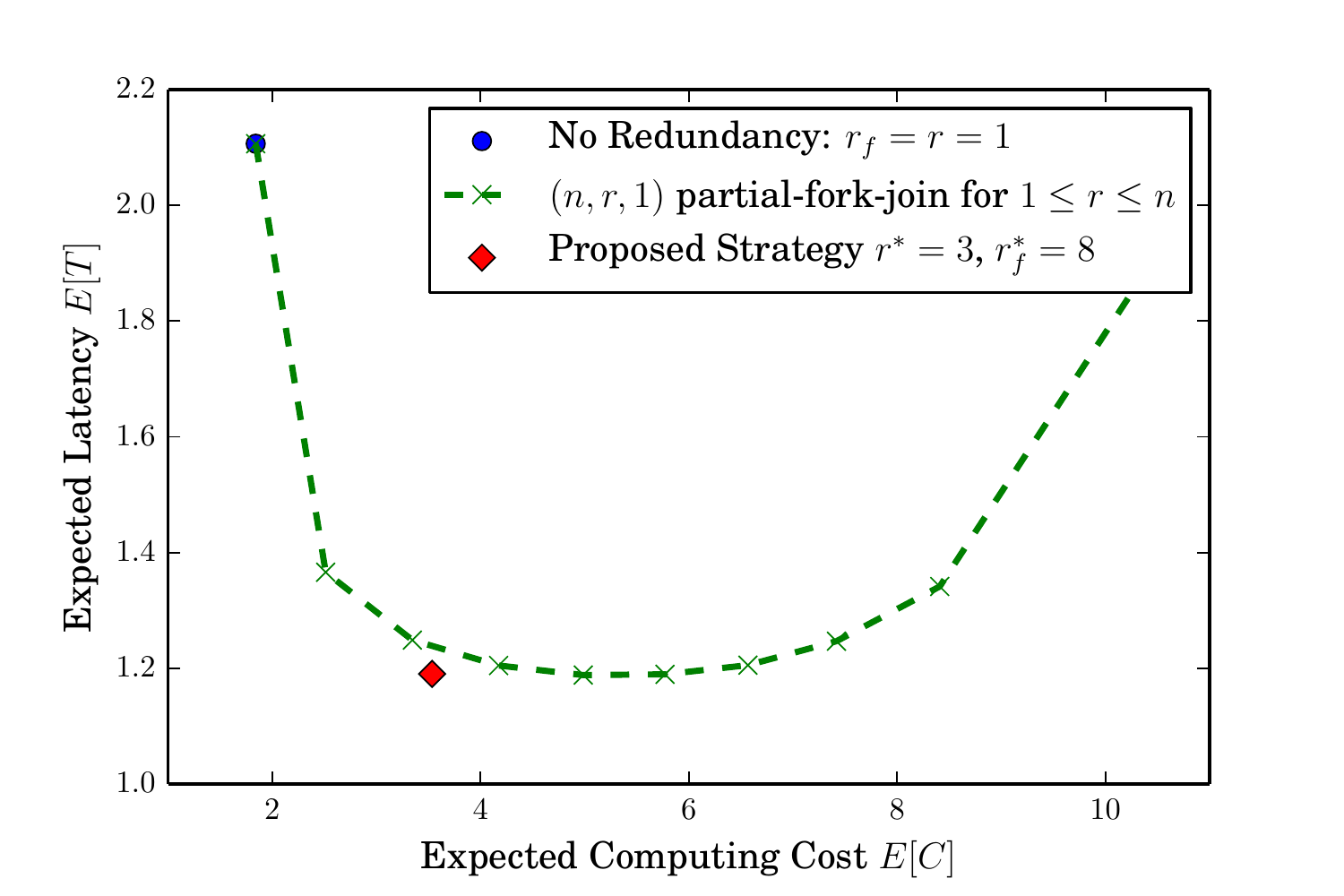}
	\caption{The latency-cost trade-off of the proposed redundancy strategy is close to that of the best $(n,r,k)$ partial-fork-join system. Service time $X \sim \Pareto(1,2.2)$, and the cost constraints are $\E{C} \leq 5$ and $r \leq  r_f \leq 8$ The first constraint is active in this example.  \label{fig:ET_vs_EC_heuristic_pareto}}
    \end{minipage}
    \hspace{0.04\linewidth}
    \begin{minipage}[t]{0.48\linewidth}
	\centering
	\includegraphics[width=3.5in]{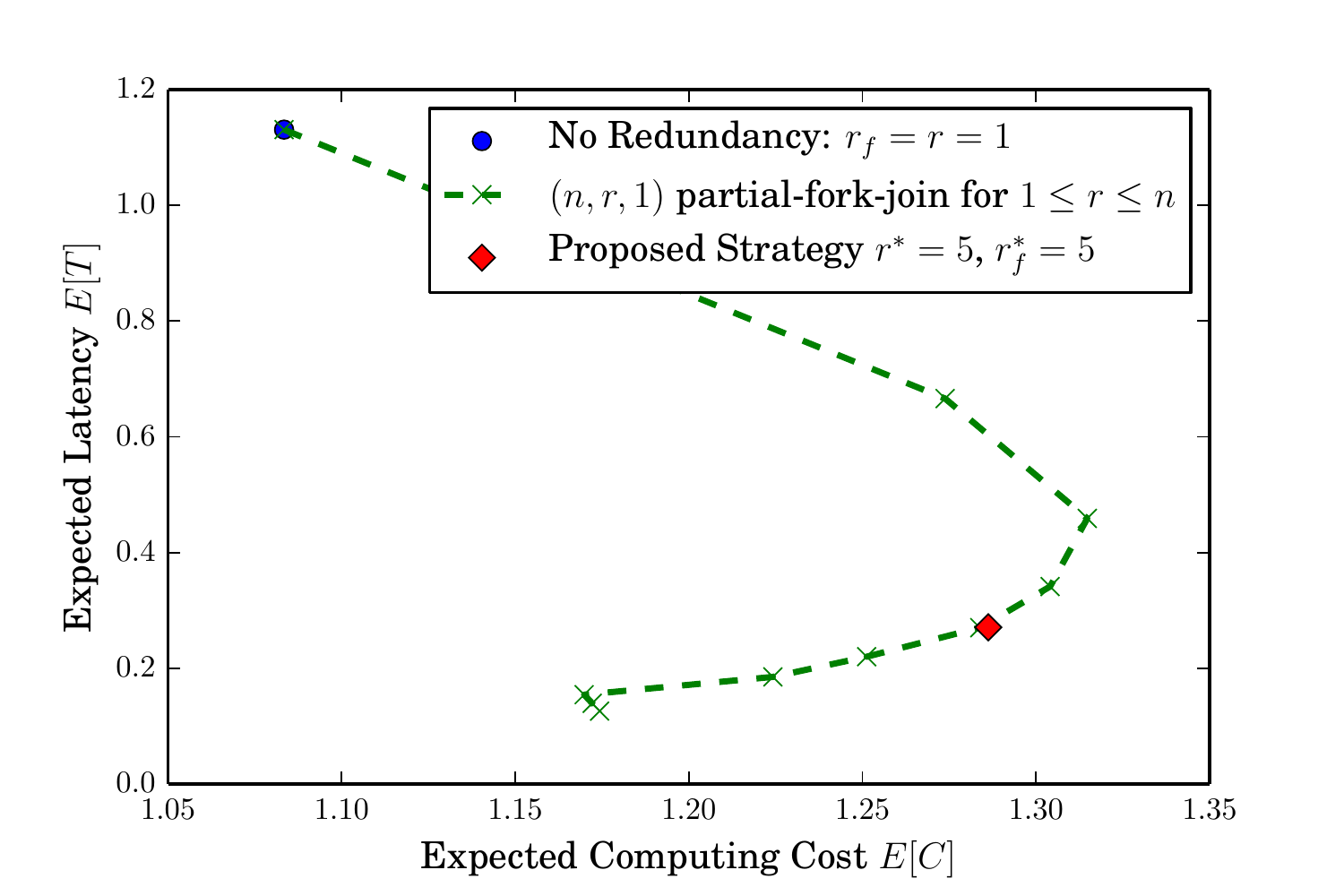}
	\caption{ The latency-cost trade-off of the proposed redundancy strategy is close to that of the best $(n,r,k)$ partial-fork-join system. The service time $X$ is an equiprobable mixture of $\Exp (2)$ and $\SExp(1,1.5)$, and the cost constraints are $\E{C} \leq 2$ and $r \leq r_f \leq 5$. The second constraint is active in this example.  \label{fig:ET_vs_EC_heuristic_mix_exp_shifted_exp}}
    \end{minipage}
\end{figure}

We now present simulation results comparing the proposed strategy given in \Cref{defn:heuristic_strategy} to the $(n,r,k)$ partial-fork-join system with $r$ varying from $k$ to $n$. The service time distributions considered here are neither log-concave nor log-convex, thus making it hard to directly infer the best redundancy strategy using the analysis presented in the previous sections. The simulations are run for $100$ workloads with $1000$ jobs each.

In Fig.~\ref{fig:ET_vs_EC_heuristic_pareto} the service time $X \sim \Pareto(1,2.2)$, $n=10$, $k=1$, and arrival rate $\lambda = 0.7$. The computing and network cost constraints are $\E{C} \leq 5$ and $r_f \leq 8$ respectively. We observe that the proposed strategy gives a significant latency reduction as compared to the no redundancy case ($r=k$ in the $(n,r,k)$ partial-fork-join system). We observe that the proposed strategy gives a latency-cost trade-off that is better than the $(n,r,k)$ partial-fork-join system. Using partial early cancellation ($r_f > r$) gives an additional reduction in latency by providing greater diversity and helping us find the $r$ out of $r_f$ queues with the least work left. % in comparison with the $(n,r,k)$ partial-fork-join system with $r. The cost $\E{C}$ increases slightly, but remains less than $\gamma$.

In Fig.~\ref{fig:ET_vs_EC_heuristic_mix_exp_shifted_exp} we show a case where the cost $\E{C}$ does not always increase with the amount of redundancy $r$. The task service time $X$ is a mixture of an exponential $\Exp(2)$ and a shifted exponential $\SExp(1,1.5)$, each occurring with equal probability. The other parameters are $n=10$, $k=1$, and arrival rate $\lambda = 0.3$. The proposed strategy found using \Cref{defn:heuristic_strategy} is $r^* = r^*_{f} = r_{max} = 5$, limited by the $r_{f} \leq r_{max}$ constraint rather than the $\E{C} \leq \gamma$ constraint. Since $r_f = r$, it coincides exactly with the $(n,r,k)$ partial-fork-join system.

\section{CONCLUDING REMARKS}
\label{sec:conclu}
In this paper we consider a redundancy model where each incoming job is forked to queues at multiple servers and we wait for any one replica to finish. We analyze how redundancy affects the latency, and the cost of computing time, and demonstrate how the log-concavity of service time is a key factor affecting the latency-cost trade-off. Some key insights from this analysis are:
\begin{itemize}
\item For log-convex service times, forking to more servers (more redundancy) reduces both latency and cost. On the other hand, for log-concave service times, more redundancy can reduce latency only at the expense of an increase in cost. 
\item Early cancellation of redundant requests can save both latency and cost for log-concave service time, but it is not effective for log-convex service time. 
%\item For log-concave service time, delayed invocation of redundant tasks helps save cost, and reduce latency in high load regime. It is ineffective for log-convex service time.
\end{itemize}

Using these insights, we also propose a general redundancy strategy for an arbitrary service time distribution, that may be neither log-concave nor log-convex. This strategy can also be used on empirical traces of service time, when a closed-form expression of the distribution is not known. % reduces latency only in the low traffic regime, and always increases the computing cost. 

Ongoing work includes developing online strategies to simultaneously learn the service time distribution, and the best redundancy strategy. More broadly, the proposed redundancy techniques can be used to reduce latency in several applications beyond the realm of cloud storage and computing systems, for example crowdsourcing, algorithmic trading, manufacturing etc.\ %data access from solid-state memory devices, etc.
%
%
%\begin{enumerate}
%\item Delayed Requests
%\item Heterogeneous servers
%\item The shift $\Delta$ being a random variable dependent on the file size, instead of a constant.
%\item If the service time distribution is not known, is it possible to estimate it and simultaneously find a good scheduling strategy?
%\end{enumerate}

\section{ACKNOWLEDGMENTS}
We thank Da Wang, Devavrat Shah, Sem Borst and Rhonda Righter for helpful dicussions. We also thank anonymous reviews for valuable feedback that helped improve this work.

\appendix
\section*{APPENDIX}
%\setcounter{section}{1}

% Appendices %
%\begin{appendices}
%
\section{LOG-CONCAVITY OF $\bar{F}_X$}
\label{sec:tail_properties}
In this section we present some properties and examples of log-concave and log-convex random variables that are relevant to this work. For more properties please see \cite{log_concave}.

\begin{propty}[Jensen's Inequality]
\label{propty:jensen_ineq}
If $\bar{F}_X$ is log-concave, then for $0 < \theta < 1$ and for all $x, y \in [0, \infty)$,
\begin{align}
\Pr(X> \theta x + (1- \theta) y) &\geq \Pr(X > x)^{\theta} \Pr(X>y)^{1-\theta}. \label{eqn:jensens}
\end{align}
The inequality is reversed if $\bar{F}_X$ is log-convex.
\end{propty}
\begin{proof}
Since $\bar{F}_X$ is log-concave, $\log \bar{F}_X$ is concave. Taking $\log$ on both sides on \eqref{eqn:jensens} we get the Jensen's inequality which holds for concave functions.
\end{proof}
In past literature saying $X$ is log-concave usually means that $f$ is log-concave. This implies that $F$ and $\bar{F}$. However log-convex $f$, does not always imply log-convexity of $F$ and $\bar{F}$.

\begin{propty}[Scaling]
\label{propty:scaling}
If $\bar{F}_X$ is log-concave, for $0 < \theta < 1$,
\begin{align}
 \Pr(X > x)  &\leq \Pr(X> \theta x)^{1/\theta} \label{eqn:scaling}
\end{align}
The inequality is reversed if $\bar{F}_X$ is log-convex.
\end{propty}

\begin{proof}
We can derive \eqref{eqn:scaling} by setting $y = 0$ in \eqref{eqn:jensens}.
\begin{align}
\Pr ( X> \theta x + (1-\theta) 0) &\geq \Pr(X>x)^{\theta} \Pr(X>0)^{1-\theta} \label{eqn:scaling_proof_1}, \\
\Pr( X >\theta x) &\geq \Pr(X>x)^{\theta} \label{eqn:scaling_proof_2}.
\end{align}
To get \eqref{eqn:scaling_proof_2} we observe that if $\bar{F}_X$ is log-concave, then $\Pr(X>0)$ has to be $1$. Otherwise log-concavity is violated at $x=0$. Raising both sides of \eqref{eqn:scaling_proof_2} to power $1/\theta$ we get \eqref{eqn:scaling}. The reverse inequality of log-convex $\bar{F}_X$ can be proved similarly.
\end{proof}

\begin{propty}[Sub-multiplicativity]
\label{propty:sub_super_additivity}
If $\bar{F}_X$ is log-concave, the conditional tail probability of $X$ satisfies for all  $t, x >0$,
\begin{align}
 &\Pr(X > x + t| X> t) \leq \Pr(X>x) \label{eqn:sub_additivity} \\
 \Leftrightarrow &\Pr(X>x+t) \leq \Pr(X>x)\Pr(X>t)
\end{align}
The inequalities above are reversed if $\bar{F}_X$ is log-convex. 
\end{propty}

\begin{proof}
\begin{align}
& \Pr(X>x) \Pr(X>t) \\
&= \Pr \left(X> \frac{x}{x+t} (x+t)\right) \Pr \left(X> \frac{t}{x+t} (x+t) \right) \label{eqn:sub_additivity_proof_1},\\
&\geq \Pr(X>x+t)^{\frac{x}{x+t}} \Pr(X>x+t)^{\frac{t}{x+t}}, \label{eqn:sub_additivity_proof_2}%\\
%&= \Pr(X>x+t)
\end{align}
where we apply \Cref{propty:scaling} to \eqref{eqn:sub_additivity_proof_1} to get \eqref{eqn:sub_additivity_proof_2}. Equation \eqref{eqn:sub_additivity} follows from \eqref{eqn:sub_additivity_proof_2}.
\end{proof}

Note that for exponential $F_X$ which is memoryless, \eqref{eqn:sub_additivity} holds with equality. Thus log-concave distributions can be thought to have `optimistic memory', because the conditional tail probability decreases over time. On the other hand, log-convex distributions have `pessimistic memory' because the conditional tail probability increases over time. The definition of the notions `new-better-than-used' in \cite{koole_righter_2008} is same as  \eqref{eqn:sub_additivity}. By \Cref{propty:sub_super_additivity} log-concavity of $\bar{F}_X$ implies that $X$ is new-better-than-used. New-better-than-used distributions are referred to as `light-everywhere' in \cite{shah_when_2013} and `new-longer-than-used' in \cite{sun_shroff}. 
\begin{propty}[Mean Residual Life]
If $\bar{F}_X$ is log-concave (log-convex), $\E{X-t |X>t}$, the mean residual life after time $t>0$ has elapsed is non-increasing (non-decreasing) in $t$.
\end{propty}

%\begin{propty}%[Expected Minimum]
%\label{propty:r_E_X_1_r_trend}
%If $X$ is log-concave (log-convex), $r \E{X_{1:r}}$ is non-decreasing (non-increasing) in $r$.
%\end{propty}

\begin{proof}[of \Cref{lem:r_E_X_1_r_trend}]
\Cref{lem:r_E_X_1_r_trend} is true for log-concave $\bar{F}_X$ if $r \E{X_{1:r}} \leq (r+1) \E{X_{1:r+1}}$ for all integers $r \geq 1$. This inequality can be simplified as follows.
\begin{align}
r \E{X_{1:r}} &\leq (r+1) \E{X_{1:r+1}} \label{eqn:r_E_X_1_r_trend_proof_6} \\
\Leftrightarrow r \int_{0}^{\infty} \Pr (X_{1:r} > x) dx &\leq \int_{0}^{\infty} (r+1) \Pr( X_{1:r+1} > x) dx , \label{eqn:r_E_X_1_r_trend_proof_5}\\
\Leftrightarrow r \int_{0}^{\infty} \Pr (X > x)^r dx &\leq \int_{0}^{\infty} (r+1) \Pr( X > x)^{r+1} dx , \label{eqn:r_E_X_1_r_trend_proof_4}\\
\Leftrightarrow \int_{0}^{\infty} \Pr \left(X > \frac{x'}{r}\right)^r dx' &\leq \int_{0}^{\infty} \Pr \left( X > \frac{x'}{r+1}\right)^{r+1} dx', \label{eqn:r_E_X_1_r_trend_proof_3}
\end{align} 

We get \eqref{eqn:r_E_X_1_r_trend_proof_5} using the fact that the expected value of a non-negative random variable is equal to the integral of its tail distribution. To get \eqref{eqn:r_E_X_1_r_trend_proof_4} observe that since $X_{1:r} = \min(X_1, X_2, \cdots , X_r)$ for i.i.d.\ $X_i$, we have $\Pr(X_{1:r} > x) = \Pr(X>x)^r$ for all $x > 0$. Similarly $\Pr(X_{1:r+1} > x) = \Pr(X>x)^{r+1}$. Next we perform a change of variables on both sides of \eqref{eqn:r_E_X_1_r_trend_proof_4} to get \eqref{eqn:r_E_X_1_r_trend_proof_3}.

Now we use \Cref{propty:scaling} to compare the two integrands in \eqref{eqn:r_E_X_1_r_trend_proof_3}. Setting $\theta = r/r+1$ and $x = x'/r$ in \Cref{propty:scaling}, we get
\begin{align}
\Pr \left(X > \frac{x'}{r} \right)^r &\leq \Pr \left( X > \frac{x'}{r+1} \right)^{r+1} \quad \text{ for all  } x' \geq 0. \label{eqn:using_scaling_propty}
\end{align}

Hence, by \eqref{eqn:using_scaling_propty} and the equivalences in \eqref{eqn:r_E_X_1_r_trend_proof_6}-\eqref{eqn:r_E_X_1_r_trend_proof_3} it follows that for log-concave $\bar{F}_X$ if $r \E{X_{1:r}}$ is non-decreasing in $r$. For log-convex $\bar{F}_X$, we can show that $r \E{X_{1:r}}$ is non-increasing in $r$ by reversing all inequalities above.

\end{proof}
%
%\begin{rem}
%If $X$ is new-better-than-used (a weaker notion implied by log-concavity of $X$), it can be shown that 
%\begin{align}
%\E{X} \leq r \E{X_{1:r}} \text{ for all integers } r \geq 1 
%\end{align}
%This is weaker than \Cref{lem:r_E_X_1_r_trend} which shows the monotonicity of $r \E{X_{1:r}}$ for log-concave (log-convex) $X$. 
%\end{rem} 

\begin{propty}[Hazard Rates]
If $\bar{F}_X$ is log-concave (log-convex), then the hazard rate $h(x)$, which is defined by $-\bar{F}'_X(x)/\bar{F}_X(x)$, is non-decreasing (non-increasing) in $x$. 
\end{propty}
%
%%\begin{propty}[Convex ordering]
%%If $X$ is log-concave then $X \leq_{cx} Exp(\mu)$, where the decay rate of the exponential, $\mu = 1/\E{X}$. The convex order inequality $X \leq_{cx} Y$ means that $\E{g(X)} \leq \E{g(Y)}$ for all convex functions $g$. The inequality of reverse for log-convex $Y$. 
%%\end{propty}
%%
\begin{propty}[Coefficient of Variation]
The coefficient of variation $C_v = \sigma/\mu$ is the ratio of the standard deviation $\sigma$ and mean $\mu$ of random variable $X$. For log-concave (log-convex) $X$, $C_v \leq 1$ ($C_v \geq 1$), and $C_v = 1$ when $X$ is pure exponential. 
\end{propty}
\begin{propty}[Examples of Log-concave $\bar{F}_X$]
The following distributions have log-concave $\bar{F}_X$:
\begin{itemize}
\item Shifted Exponential (Exponential plus constant $\Delta > 0$)
\item Uniform over any convex set
\item Weibull with shape parameter $c \geq 1$
\item Gamma with shape parameter $c \geq 1$
\item Chi-squared with degrees of freedom $c \geq 2$
\end{itemize}
\end{propty}

\begin{propty}[Examples of Log-convex $\bar{F}_X$]
The following distributions have log-convex $\bar{F}_X$:
\begin{itemize}
\item Exponential
\item Hyper Exponential (Mixture of exponentials)
\item Weibull with shape parameter $0 < c <1$
\item Gamma with shape parameter $0 < c < 1$
\end{itemize}
\end{propty}

\section{PROOFS FOR THE $k=1$ Case}
\label{sec:rep_queueing_proofs}

\begin{proof}[of \Cref{thm:E_C_r_trend}]
Using \eqref{eqn:C_expr}, we can express the cost $C$ in terms of the relative task start times $t_i$, and $S$ as follows. Since only $r$ tasks are invoked, the relative start times $t_{r+1}, \dots , t_n$ are equal to $\infty$. %The computing cost can be expressed in terms of $S$ as,
\begin{align}
C &= S + \posfunc{ S - t_2} + \cdots + \posfunc{S - t_r} \label{eqn:C_expr_r},
\end{align}
where $S$ is the time between the start of service of the earliest task, and when any $1$ of the $r$ tasks finishes. The tail distribution of $S$ is given by
\begin{align}
\Pr(S> s) &= \prod_{i=1}^{r} \Pr(X > s - t_i).
\end{align}
By taking expectation on both sides of \eqref{eqn:C_expr_r} and simplifying we get,
\begin{align}
\E{C} &= \sum_{u=1}^{r} \int_{t_u}^{\infty}  \Pr(S>s) ds, \\
&= \sum_{u=1}^{r} u \int_{t_u}^{t_{u+1}} \Pr(S>s) ds, \\
&= \sum_{u=1}^{r} u \int_{0}^{t_{u+1}-t_u} \Pr(S > t_u + x) dx ,\\
&=  \sum_{u=1}^{r}  u \int_{0}^{t_{u+1}-t_u} \prod_{i=1}^{u} \Pr(X > x+t_u-t_i) dx .\label{eqn:E_C_simp}
\end{align}

We now prove that for log-concave $\bar{F}_X$, $\E{C} \geq \E{X}$. The proof that $\E{C} \leq \E{X}$ when $\bar{F}_X$ is log-convex follows similarly with all inequalities below reversed. We express the integral in \eqref{eqn:E_C_simp} as,
\begin{align}
\E{C} &= \sum_{u=1}^{r}  u \left(\int_{0}^{\infty} \prod_{i=1}^{u} \Pr(X > x+t_u-t_i) dx - \int_{0}^{\infty} \prod_{i=1}^{u} \Pr(X > x+t_{u+1}-t_i) dx \right), \label{eqn:E_C_E_X_proof_1} \\
&= \sum_{u=1}^{r} \left( \int_{0}^{\infty} \prod_{i=1}^{u} \Pr \left(X > \frac{x'}{u}+t_u-t_i \right) dx' - \int_{0}^{\infty} \prod_{i=1}^{u} \Pr \left(X > \frac{x'}{u}+t_{u+1}-t_i \right) dx'
 \right) , \label{eqn:E_C_E_X_proof_2}\\
 &= \E{X} + \sum_{u=2}^{r} \int_{0}^{\infty} \left( \prod_{i=1}^{u} \Pr \left(X > \frac{x'}{u}+t_u-t_i \right) - \prod_{i=1}^{u-1} \Pr \left(X > \frac{x'}{u-1}+t_{u}-t_i \right) \right) dx' \label{eqn:E_C_E_X_proof_3},\\
&\geq \E{X},
\end{align}

where in \eqref{eqn:E_C_E_X_proof_1} we express each integral in \eqref{eqn:E_C_simp} as a difference of two integrals from $0$ to $\infty$. In \eqref{eqn:E_C_E_X_proof_2} we perform a change of variables $x = x'/u$. In \eqref{eqn:E_C_E_X_proof_3} we rearrange the grouping of the terms in the sum; the $u^{th}$ negative integral is put in the $u+1$ term of the summation. Then the first term of the summation is simply  $\int_{0}^{\infty} \Pr(X>x) dx$ which is equal to $\E{X}$. In \eqref{eqn:E_C_E_X_proof_3} we use the fact that each term in the summation in \eqref{eqn:E_C_E_X_proof_2} is positive when $\bar{F}_X$ is log-concave. This is shown in Lemma~\ref{lem:E_C_E_X_lemma} below. 

Next we prove that for log-concave $\bar{F}_X$, $\E{C} \leq r \E{X_{1:r}}$. Again, the proof of $\E{C} \geq r \E{X_{1:r}}$ when $\bar{F}_X$ is log-convex follows with all the inequalities below reversed.

\begin{align}
\E{C} &\leq \sum_{u=1}^{r} u \int_{0}^{t_{u+1}-t_u} \prod_{i=1}^{u}\Pr \left(X > \frac{u(x+t_u-t_i)}{r} \right) ^{r/u} dx \label{eqn:E_C_r_E_X_1_r_proof_1},\\
&= \sum_{u=1}^{r} \left( \int_{0}^{\infty} \prod_{i=1}^{u}\Pr \left(X > \frac{x'+u(t_u-t_i)}{r} \right) ^{r/u} dx' - \int_{0}^{\infty} \prod_{i=1}^{u}\Pr \left(X > \frac{x'+u(t_{u+1}-t_i)}{r} \right) ^{r/u} dx' \right) , \label{eqn:E_C_r_E_X_1_r_proof_2}\\
&= \int_{0}^{\infty} \Pr \left( X > \frac{x'}{r} \right)^r dx' + \sum_{u=2}^{r} \left( \int_{0}^{\infty} \prod_{i=1}^{u}\Pr \left(X > \frac{x' + u(t_u-t_i)}{r} \right) ^{r/u} dx'  - \right. \nonumber \\ 
&\hspace{ 5.5 cm} \left. \int_{0}^{\infty} \prod_{i=1}^{u-1} \Pr \left(X > \frac{x'+ (u-1) (t_{u}-t_i)}{r} \right)^{\frac{r}{u-1}}dx'
 \right), \label{eqn:E_C_r_E_X_1_r_proof_3}\\
 &\leq r \E{X_{1:r}},
\end{align}
where we get  \eqref{eqn:E_C_r_E_X_1_r_proof_1} by applying \Cref{propty:scaling} to \eqref{eqn:E_C_simp}. In \eqref{eqn:E_C_r_E_X_1_r_proof_2} we express the integral as a difference of two integrals from $0$ to $\infty$, and perform a change of variables $x = x'/u$. In \eqref{eqn:E_C_r_E_X_1_r_proof_3} we rearrange the grouping of the terms in the sum; the $u^{th}$ negative integral is put in the $u+1$ term of the summation. The first term is equal to $r \E{X_{1:r}}$. We use \Cref{lem:E_C_r_E_X_1_r_lemma} to show that each term in the summation in \eqref{eqn:E_C_r_E_X_1_r_proof_3} is negative when $\bar{F}_X$ is log-concave.
\end{proof}

\begin{lem}
\label{lem:E_C_E_X_lemma}
If $\bar{F}_X$ is log-concave,
\begin{align}
& \prod_{i=1}^{u} \Pr \left(X > \frac{x'}{u}+t_u-t_i \right) \geq \prod_{i=1}^{u-1} \Pr \left(X > \frac{x'}{u-1}+t_{u}-t_i \right).
\end{align}
The inequality is reversed for log-convex $\bar{F}_X$. 
\end{lem}

\begin{proof}[of \Cref{lem:E_C_E_X_lemma}]
We bound the left hand side expression as follows.
\begin{align}
\prod_{i=1}^{u} \Pr \left(X > \frac{x}{u}+t_u-t_i \right) &= \Pr(S>t_u) \prod_{i=1}^{u} \Pr \left(X > \frac{x}{u}+t_u-t_i | X > t_u-t_i\right), \label{eqn:E_C_E_X_lemma_proof_1}\\
&= \Pr(S>t_u)  \Pr \left(X> \frac{x}{u} \right)^{\frac{u-1}{u-1}} \times \prod_{i=1}^{u-1} \Pr \left(X > \frac{x}{u}+t_u-t_i | X > t_u-t_i\right), \label{eqn:E_C_E_X_lemma_proof_2}\\
&\geq \Pr(S>t_u) \prod_{i=1}^{u-1} \Pr \left(X > \frac{x}{u}+t_u-t_i | X > t_u-t_i\right)^{\frac{u}{u-1}}, \label{eqn:E_C_E_X_lemma_proof_3} \\
&\geq \Pr(S >t_u) \prod_{i=1}^{u-1} \Pr \left(X > \frac{x}{u-1}+t_u-t_i | X > t_u-t_i\right), \label{eqn:E_C_E_X_lemma_proof_4} \\
& = \prod_{i=1}^{u-1} \Pr \left(X > \frac{x}{u-1}+t_{u}-t_i \right)
\end{align}
where we use Property~\ref{propty:sub_super_additivity} to get \eqref{eqn:E_C_E_X_lemma_proof_3}. The inequality in \eqref{eqn:E_C_E_X_lemma_proof_4} follows from applying \Cref{propty:scaling} to the conditional distribution $\Pr(Y > x'/u) = \Pr(X > x'/u +t_u-t_i | X > t_u-t_i )$, which is also log-concave. 

For log-convex $\bar{F}_X$ all the inequalities can be reversed. 
\end{proof}

\begin{lem}
\label{lem:E_C_r_E_X_1_r_lemma}
If $\bar{F}_X$ is log-concave,
\begin{align}
& \prod_{i=1}^{u}\Pr \left(X > \frac{x + u(t_u-t_i)}{r} \right) ^{r/u}  \leq  \prod_{i=1}^{u-1} \Pr \left(X > \frac{x+ (u-1) (t_{u}-t_i)}{r} \right)^{\frac{r}{u-1}}
\end{align}
The inequality is reversed for log-convex $\bar{F}_X$. 
\end{lem}

\begin{proof}[of \Cref{lem:E_C_r_E_X_1_r_lemma}]
We start by simplifying the left-hand side expression, raised to the power $(u-1)/r$.
\begin{align}
\prod_{i=1}^{u}\Pr \left(X > \frac{x + u(t_u-t_i)}{r} \right) ^{(u-1)/u} &= \Pr \left( X > \frac{x}{r} \right)^{\frac{u-1}{u}} \prod_{i=1}^{u-1} \Pr \left(X > \frac{x+ u (t_{u}-t_i)}{r} \right)^{\frac{u-1}{u}} \\
&=  \prod_{i=1}^{u-1} \Pr \left( X > \frac{x}{r} \right)^{\frac{1}{u}}  \Pr \left(X > \frac{x+ u (t_{u}-t_i)}{r} \right)^{\frac{u-1}{u}} \\
&\leq \prod_{i=1}^{u-1}  \Pr \left(X > \frac{x+ (u-1) (t_{u}-t_i)}{r} \right) \label{eqn:E_C_r_E_X_1_r_lemma_proof}
\end{align}
where \eqref{eqn:E_C_r_E_X_1_r_lemma_proof} follows from the log-concavity of $\Pr(X>x)$, and the Jensen's equality. The inequality is reversed for log-convex $\bar{F}_X$. 
\end{proof}
\section{PROOFS FOR GENERAL $k$}
\label{sec:coded_queueing_proofs}
\begin{proof}[of \Cref{thm:latency_bnds_gen}]
To find the upper bound on latency, we consider a related queueing system called the split-merge queueing system. In the split-merge system all the queues are blocked and cannot serve subsequent jobs until $k$ out of $n$ tasks of the current job are complete. Thus the latency of the split-merge system serves as an upper bound on that of the fork-join system. In the split-merge system we observe that jobs are served one-by-one, and no two jobs are served simultaneously. So it is equivalent to an $M/G/1$ queue with Poisson arrival rate $\lambda$, and service time $X_{k:n}$. The expected latency of an $M/G/1$ queue is given by the Pollaczek-Khinchine formula \cite[Chapter~5]{dsp_gallager}, and it reduces to the upper bound in \eqref{eqn:upper_bnd_gen}.

To find the lower bound we consider a system where the job requires $k$ out of $n$ tasks to complete, but all jobs arriving before it require only $1$ task to finish. Then the expected waiting time in queue is equal to the second term in \eqref{eqn:upper_bnd_gen} with $k$ set to $1$. Adding the expected service time $\E{X_{k:n}}$ to this lower bound on expected waiting time, we get the lower bound \eqref{eqn:lower_bnd_gen} on the expected latency.
\end{proof}

\begin{proof}[of \Cref{lem:tighter_upper_bnd}]
This upper bound is a generalization of the bound on the mean response time of the $(n,n)$ fork-join system with exponential service time presented in \cite{nelson_tantawi}. To find the bound, we first observe that the response times experienced by the tasks in the $n$ queues form a set of associated random variables \cite{assoc_rand_vars}. Then we use the property of associated random variables that their expected maximum is less than that for independent variables with the same marginal distributions. Unfortunately, this approach cannot be extended to the $k < n$ case because this property of associated variables does not hold for the $k^{th}$ order statistic for $k<n$.
\end{proof}

\begin{proof}[of \Cref{thm:cost_bounds_gen}]
A key observation used in proving the cost bounds is that at least $n-k+1$ out of the $n$ tasks of a job $i$ start service at the same time. This is because when the $k^{th}$ task of Job $(i-1)$ finishes, the remaining $n-k$ tasks are canceled immediately. These $n-k+1$ queues start working on the tasks of Job $i$ at the same time. %Note that the result holds trivially if job $i$ arrives when all $n$ queues are idle. 

To prove the upper bound we divide the $n$ tasks into two groups, the $k-1$ tasks that can start early, and the $n-k+1$ which start at the same time after the last tasks of the previous job are terminated. We consider a constraint that all the $k-1$ tasks in the first group and $1$ of the remaining $n-k+1$ tasks needs to be served for completion of the job. This gives an upper bound on the computing cost because we are not taking into account the case where more than one tasks from the second group can finish service before the $k-1$ tasks in the first group. For the $n-k+1$ tasks in the second group, the computing cost is equal to $n-k+1$ times the time taken for one of them to complete. The computing time spent on the first $k-1$ tasks is at most $(k-1) \E{X}$. Adding this to the second group's cost, we get the upper bound \eqref{eqn:E_C_gen_upper_bnd}.

We observe that the expected computing cost for the $k$ tasks that finish is at least $ \sum_{i=1}^{k} \E{X_{i:n}}$, which takes into account full diversity of the redundant tasks. Since we need $k$ tasks to complete in total, at least $1$ of the $n-k+1$ tasks that start simultaneously needs to be served. Thus, the computing cost of the $(n-k)$ redundant tasks is at least $(n-k) \E{X_{1:n-k+1}}$. Adding this to the lower bound on the first group's cost, we get \eqref{eqn:E_C_gen_lower_bnd}.
\end{proof}

\begin{proof}[of \Cref{thm:early_cancel_gen}]
Since exactly $k$ tasks are served, and others are cancelled before they start service, it follows that the expected computing cost $\E{C} = k \E{X}$. In the sequel, we find an upper bound on the latency of the $(n,k)$ fork-early-cancel system. 

First observe that in the $(n,k)$ fork-early-cancel system, the $n-k$ redundant tasks that are canceled early help find the $k$ shortest queues. The expected task arrival rate at each server is $\lambda k/n$, which excludes the redundant tasks that are canceled before they start service.

Consider an $(n,k,k)$ partial fork system without redundancy, where the $k$ tasks of each job are assigned to $k$ out of $n$ queues chosen uniformly at random. The job exits the system when all $k$ tasks are complete. The expected task arrival rate at each server is $\lambda k/n$, same as the $(n,k)$ fork-early-cancel system. However, the $(n,k)$ fork-early-cancel system gives lower latency because having the $n-k$ redundant tasks provides diversity and helps find the $k$ shortest queues. Thus the latency of the $(n,k,k)$ partial-fork-join system is bounded below by that of the $(n,k)$ fork-early-cancel system. 

Now let us upper bound the latency $\E{T^{(pf)}}$ of the partial fork system. Each queue has arrival rate $\lambda k/n$,  and service time distribution $F_X$. Using the approach in \cite{nelson_tantawi} we can show that the response times (waiting plus service time) $R_i$, $ 1 \leq i \leq k$ of the $k$ queues serving each job form a set of associated random variables. Then by the property that the expected maximum of $k$ associated random variables is less than the expected maximum of $k$ independent variables with the same marginal distributions we can show that,
\begin{align}
\E{T} &\leq \E{T^{(pf)}} \\
&\leq \E{\max \left( R_1, R_2, \cdots R_k \right)}.
\end{align}
The expected maximum can be numerically evaluated from distribution of $R$. From the transform analysis given in \cite[Chapter~25]{mor_book}, we know that the Laplace-Stieltjes transform $R(s)$ of the probability density of $R$ is same as \eqref{eqn:R_i_pdf_laplace_tx}, but with $\lambda$ replaced by $\lambda k/n$.
\end{proof}

%
%\end{appendices}

%  Bibliography %

\bibliographystyle{acmlarge}
\bibliography{BibTex/storage,BibTex/computing}

\end{document}